\documentclass[11pt]{article}
\usepackage{epsfig,amsmath,amsthm,latexsym,graphicx,amsfonts,amssymb}
\usepackage[numbers]{natbib}
\newcommand{\beq}{\begin{equation}}
\newcommand{\eeq}{\end{equation}}
\newcommand{\R}{\mathbb R}

\newcommand{\E}{\mathbb E}

\newcommand{\cX}{\mathcal{X}}
\newcommand{\cM}{\mathcal{M}}

\newcommand{\htheta}{\hat{\theta}}
\newcommand{\hxi}{\hat{\xi}}

\newcommand{\ac}[1]{\mathrm{AC}^{(#1)}}

\newcommand{\vol}{\operatorname{\mathrm{Vol}}}
\newcommand{\var}{\operatorname{\mathrm{Var}}}
\newcommand{\comp}{\operatorname{\mathrm{Comp}}}

\newcommand{\phixi}{\operatorname{{\Phi_\xi}}}
\newcommand{\phix}{\operatorname{{\Phi_x}}}
\newcommand{\defeq}{\stackrel{\mbox{\tiny{def}}}{=}}
\newtheorem{theorem}{Theorem}

\newtheorem{lemma}[theorem]{Lemma}

\newtheorem{remark}{Remark}
\newtheorem{regcond}{Condition}

\begin{document}

\title{Higher-order asymptotics for the parametric complexity}
\author{James G. Dowty}
\date{\today}

\maketitle

\abstract{
The parametric complexity is the key quantity in the minimum description length (MDL) approach to statistical model selection.  Rissanen and others have shown that the parametric complexity of a statistical model approaches a simple function of the Fisher information volume of the model as the sample size $n$ goes to infinity.  This paper derives higher-order asymptotic expansions for the parametric complexity, in the case of exponential families and independent and identically distributed data.  These higher-order approximations are calculated for some examples and are shown to have better finite-sample behaviour than Rissanen's approximation.  The higher-order terms are given as expressions involving cumulants (or, more naturally, the Amari-Chentsov tensors), and these terms are likely to be interesting in themselves since they arise naturally from the general information-theoretic principles underpinning MDL.
The derivation given here specializes to an alternative and arguably simpler proof of Rissanen's result (for the case considered here), proving for the first time that his approximation is $O(n^{-1})$.
}

\section{Introduction}
\label{S:intro}

The minimum description length (MDL) principle provides a general information-theoretic approach to model selection and other forms of statistical inference \citep{BarronEtAl98,Rissanen07}.  The MDL criterion for model selection is consistent, meaning that it will select the data-generating model from a countable set of competing parametric models with probability approaching $1$ as the sample size $n$ goes to infinity \citep{BarronCover91}.  For example, if each of the parametric models is a logistic regression model with predictor variables taken from a fixed set of potential predictors, then the MDL model-selection criterion will choose the correct combination of predictors with probability approaching $1$ as $n \to \infty$.  The MDL model-selection criterion also has a number of strong optimality properties, which greatly extend Shannon's noiseless coding theorem \citep[\S III.E]{BarronEtAl98}.

In its simplest form, the MDL principle advocates choosing the model for which the observed data has the shortest message length under a particular prefix code defined by a minimax condition
\citep[\S 2.4.3]{Grunwald05}. \citet{Shtarkov87} showed that this is equivalent to choosing the model with the largest normalized maximum likelihood (NML) for the observed data.  Here, minus the logarithm of the NML for a model $\cM$ and $n$ observations $x^n$ is
\begin{align*}
- \log p_n(x^n, \htheta(x^n)) + \comp
\end{align*}
where: $p_n(x^n, \theta)$ is the likelihood function for data $x^n$ at model parameter $\theta$; $\htheta(x^n)$ is the maximum likelihood estimate of $\theta$ corresponding to $x^n$; and the parametric complexity of the model $\cM$ is
\begin{align}
\label{E:pc_defn_general}
\comp = \log \int p_n(x^n, \htheta(x^n)) \; d\mu_n(x^n),
\end{align}
where the integral is over all possible values of the data $x^n$ (and is technically the Lebesgue integral with respect to some measure $\mu_n$, so the integral becomes a sum if $\mu_n$ is discrete).

The NML balances how well the model fits the data (quantified by the maximized likelihood) against how complex or versatile the model is (quantified by the parametric complexity), so the NML gives a natural measure of parsimony.  However, in many cases of interest, it is not practical to calculate the parametric complexity directly from the definition (\ref{E:pc_defn_general}).  For example, the data space for a logistic regression model consists of all binary sequences $x^n$ of length $n$, and the parametric complexity is defined in terms of the sum of $p_n(x^n, \htheta(x^n))$ over this space, so calculating this sum directly is infeasible even for $n = 100$ since there are $2^{100} \approx 1.27 \times 10^{30}$ terms.  \citet{ClarkeBarron90,ClarkeBarron94} therefore obtained $o(1)$ approximations to the parametric complexity in the limit $n \to \infty$ for independent and identically distributed (IID) data, and \citet{Rissanen96} greatly extended these results to the non-IID setting.

Approximations to $\comp$ are only sensible when $\comp$ is finite, but this condition fails in many cases.  In proving their $o(1)$ approximation,  \citet{ClarkeBarron90,ClarkeBarron94} and \citet{Rissanen96} therefore effectively restricted the data to those data points $x^n$ whose corresponding maximum likelihood estimates $\htheta(x^n)$ lie in a given compact (i.e., closed and bounded) subset $K$ of the parameter space $\Theta$.
Restricting $x^n$ to this set of data points in the integral (\ref{E:pc_defn_general}) gives a quantity $\comp(K)$ which is always finite in our main case of interest.
Then the approximation of \citep{ClarkeBarron90,ClarkeBarron94,Rissanen96} is
\beq \label{E:PC_volapprox}
\comp(K) =  \frac{d}{2}\log \frac{n}{2\pi} + \log \int_K d\pi_\Theta(\theta) + o(1)
\eeq
as $n \to \infty$, where $d$ is the dimension of $\Theta$, $\pi_\Theta$ is the Jeffreys prior on $\Theta$ (scaled so it does not depend on $n$) and $\int_K d\pi_\Theta(\theta) = \vol(K)$ is the Fisher information volume of $K$.

Suppose now that the data are IID and that $\cM$ is a natural exponential family, and let $\Theta \subseteq \R^d$ be any parameter space for $\cM$, not necessarily the natural parameterisation.  Assume that each distribution in $\cM$ satisfies Cram\'er's condition (which holds for all continuous distributions) and that $K \subseteq \Theta$ is a compact, $d$-dimensional subset of $\R^d$ with smooth boundary (see the end of Section \ref{S:defns} for precise statements of these conditions).  Then our main result is the following refinement of (\ref{E:PC_volapprox}), for any non-negative integer $s$:
\begin{align}
\label{E:PC_fullapprox}
\comp(K) = \frac{d}{2}\log \frac{n}{2\pi} + \log \int_K \left( \sum_{i=0}^s  F_i(\theta) n^{-i} \right) d\pi_\Theta(\theta) + o(n^{-s})
\end{align}
as $n \to \infty$, where each $F_i(\theta)$ is a smooth function which does not depend on $n$.

Our methods can be used to calculate the asymptotic expansion (\ref{E:PC_fullapprox}) for $\comp(K)$ to any desired degree of accuracy.  The first few of the functions $F_i(\theta)$ are given explicitly in Theorem \ref{T:Edgeworth_fn} in terms of the Fisher information metric and cumulants.  In particular, $F_0(\theta)=1$, so taking $s=0$ in (\ref{E:PC_fullapprox}) gives the approximation (\ref{E:PC_volapprox}).  Further, setting $s=1$ in (\ref{E:PC_fullapprox}) shows, at least for the case of exponential families and IID data, that the approximation (\ref{E:PC_volapprox}) is actually valid to order $O(n^{-1})$, which appears to be previously unknown.

The functions $F_i(\theta)$ which appear in (\ref{E:PC_fullapprox}) are likely to be interesting, since they arise naturally from the general, information-theoretic principles which underpin MDL.  These functions appear to be more related to the statistical aspects of the model $\cM$ than to its curvature or other geometrical aspects, see Remark \ref{R:geom}.  Theorem \ref{T:Edgeworth_fn} gives these functions in terms of the Fisher information metric and cumulants, but a more natural formulation in terms of the Amari-Chentsov tensors is also given in Section \ref{S:AC}.

The rest of this paper is set out as follows.  In Section \ref{S:defns} we define the main objects of this paper and fully describe the regularity conditions under which (\ref{E:PC_fullapprox}) holds.  In Section \ref{S:higherorder} we give an asymptotic expansion for the parametric complexity in terms of a function which arises in Edgeworth expansions.  Explicit formulae for the first few terms of this function are then obtained in Section \ref{S:integrand}, completing the proof of (\ref{E:PC_fullapprox}).  In Section \ref{S:examples}, the first two terms of (\ref{E:PC_fullapprox}) are calculated explicitly for some examples.  In Section \ref{S:finite_sample}, the finite-sample performance of (\ref{E:PC_volapprox}) and these higher order approximations are assessed, using an exact formula for the parametric complexity in the case of spherical normal data.  We give alternative, co-ordinate-independent formulae for the functions $F_i(\theta)$ in Section \ref{S:AC}, before finishing with a summary of the paper in Section \ref{S:conclusion}.  Appendix \ref{S:proofs} contains the proofs of all of our results and Appendix \ref{S:appendixB} gives a fairly self-contained description of Edgeworth expansions and the Hermite numbers.

\section{Definitions and regularity conditions}
\label{S:defns}

As in the Introduction, let $\cM$ be a natural exponential family of order $d$ \citep[\S 2.2.1]{KassVos97} and let $\Theta$ be any parameter space of $\cM$, which we assume without loss of generality is an open subset of $\R^d$.  Then $\cM$ is a set of probability measures on $\R^d$, with each distribution in $\cM$ of the form $p_1(\cdot,\theta) \mu_1$ for some $\theta \in \Theta$, where: $\mu_1$ is some measure on $\R^d$; $p_1(\cdot,\theta)$ is the function $x \mapsto p_1(x,\theta)$ of $x$ alone, with $\theta$ held fixed;
$$p_1(x,\theta) = \exp(x \cdot \eta(\theta) - \psi(\eta(\theta)));$$
$\eta: \Theta \to \R^d$ is the reparameterisation map from $\Theta$ to the natural parameter space (a diffeomorphism onto its image); $\psi$ is the log-partition function, which is determined by the condition that each measure $p_1(\cdot,\theta) \mu_1$ is normalised; and $x \cdot \eta(\theta)$ is the Euclidean inner product of $x$ and $\eta(\theta)$.

Let $\cM^n$ be the corresponding model for $n$ IID observations $x_1, \dots, x_n \in \R^d$ and let $x^n = (x_1, \dots, x_n) \in \R^{nd}$ be the observed data for this model.  Then $\cM^n$ is the set of all measures of the form $p_n(\cdot, \theta) \mu_n$ for $\theta \in \Theta$, where
$p_n(x^n, \theta) = p_1(x_1, \theta) \dots p_1(x_n, \theta)$ is the product likelihood and $\mu_n = \mu_1 \times \dots \times \mu_1$ is the product measure on $\R^{nd}$.

The (scaled) Fisher information metric $g_\Theta$ on $\Theta$ is the $d \times d$ matrix-valued function on $\Theta$ with $(i,j)^{th}$ entry
\begin{align}
\label{E:FI_defn}
g_{ij}(\theta) =
- \frac{1}{n} \int \left[\frac{\partial^2 \log p_n(x^n, \theta) }{\partial \theta_i \partial \theta_j}  \right]
p_n(x^n, \theta) \; d\mu_n(x^n)
\end{align}
for $i,j = 1, \dots, d$, where the factor of $1/n$ ensures that $g_\Theta$ does not depend on $n$ (since the data are IID).  Then the Jeffreys prior $\pi_\Theta$ is $\sqrt{\det g_\Theta}$ times the Lebesgue measure on $\R^d$, and so is equal to the Fisher information volume density on $\Theta$.

Let $\cX^n$ be the support of $\mu_n$, which can be interpreted as the set of all possible values of the observed data $x^n$.  Let $\htheta: \cX^n \to \Theta$ be the maximum likelihood estimator for the parameter $\theta$, so that $\htheta(x^n) \in \Theta$ is the maximum likelihood estimate corresponding to $x^n$.  This estimate does not exist for all $x^n$ but it is unique when it does exist, so $\htheta$ is well-defined if we restrict $\cX^n$ appropriately \citep[Corollary 9.6]{BarndorffNielsen78}.
If $K$ is any compact subset of $\Theta$ then let
$$\htheta^{-1}(K) = \{ x^n \in \cX^n \mid \htheta(x^n) \in K \}$$
be the set of data points $x^n$ whose maximum likelihood estimates exist and lie in $K$.  Then define
\begin{align}
\label{E:pc_defn_K}
\comp(K) = \log \int_{\htheta^{-1}(K)} p_n(x^n, \htheta(x^n)) \; d\mu_n(x^n)
\end{align}
to be the contribution to the parametric complexity corresponding to $K \subseteq \Theta$.

Now, the maximum likelihood estimate $t = \htheta(x^n) \in \Theta$ is a random variable, so write its distribution as
$q_n(\cdot,\theta) \nu_n$,
where $q_n(t,\theta)$ is some smooth function and $\nu_n$ is some measure on $\R^d$ which is independent of $\theta$.
In fact, the family of distributions $q_n(\cdot,\theta) \nu_n$ for all $\theta \in \Theta$ is also an exponential family, though we will not need this result here.  For now, we only use the fact that the maximum likelihood estimator is a sufficient statistic (by Theorem 2.2.6 of \citep{KassVos97} and the paragraph preceding it) and that the parametric complexity can be calculated from the distribution of any sufficient statistic \citep[\S III.F]{BarronEtAl98}.  Then (\ref{E:pc_defn_K}) and the argument in \citep[\S III.F]{BarronEtAl98} imply
\begin{align}
\label{E:PC_defn}
\comp(K) = \log \int_K  q_n(t,t) \; d\nu_n(t),
\end{align}
which is intuitively reasonable since $q_n(t,t)$ is the maximized likelihood.

\begin{remark}
\label{R:proofidea}
\citet{Rissanen96} proved (\ref{E:PC_volapprox}) by using the central limit theorem to approximate the integrand of (\ref{E:PC_defn}), despite the fact that the central limit theorem describes the distribution of $t$ for any fixed $\theta$
while the integrand of (\ref{E:PC_defn}) is $q_n(t,\theta(t))$, where $\theta(t)=t$ is not fixed.  To get around this problem, Rissanen effectively replaced the function $\theta(t)$ by a step function $\bar{\theta}(t)$ which closely approximates $\theta(t)$, then he applied the central limit theorem in each of the small sets where $\bar{\theta}(t)$ is constant (and hence $\theta$ is fixed).
By contrast, our approach will be to approximate the integral (\ref{E:PC_defn}) by (a multiple of) the double integral of $q_n(t, \theta)$ with respect to $t$ and $\theta$, where $(t,\theta)$ ranges over a small neighbourhood of the set
$\{ (t,\theta) \mid \theta = t \mbox{ and } t \in K \}$.
If we integrate with respect to $t$ first and $\theta$ second in this double integral then $\theta$ is fixed in the innermost integral, so we can apply the central limit theorem (and higher-order Edgeworth expansions) there.
\end{remark}

We now impose the following regularity conditions, which will be needed when we use Edgeworth expansions in the next section to largely prove (\ref{E:PC_fullapprox}).

\begin{regcond}
We assume that each distribution $Q$ in the exponential family $\cM$ satisfies Cram\'er's condition
\begin{align}
\label{E:condCramer}
\limsup_{\| z\| \to \infty} |\hat{Q}(z)| < 1,
\end{align}
where $\hat{Q}$ is the characteristic function of $Q$ \citep[eqn. 1.29]{BhattacharyaDenker90}.
Cram\'er's condition is satisfied for many interesting distributions, such as all continuous distributions (those which admit probability density functions) \citep[p. 57]{Hall92} and the distributions of the minimal sufficient statistics of continuous exponential families \citep[Lemma 1.5]{BhattacharyaDenker90}.
\end{regcond}

\begin{regcond}
\label{C:condK}
In addition to assuming that $K \subseteq \Theta \subseteq \R^d$ is compact, we assume that $K$ is a smooth $d$-manifold with boundary.  This means that each point of $K$ is at the centre of some small $d$-dimensional ball which intersects $K$ either in the whole ball (for interior points of $K$) or in approximately a half ball (for boundary points of $K$).
\end{regcond}

\section{Higher order asymptotics}
\label{S:higherorder}

In this section, we will mostly work in the expectation parameter space of the exponential family $\cM$, though the details of this parameter space are not important here beyond two specific properties, which we now recall.
Let $\Xi$ denote the expectation parameter space and let $\xi \in \Xi$ denote a generic expectation parameter (which is equal to the expected value of the sufficient statistic).  If $x_1, \dots, x_n$ are IID random variables governed by an unknown element of $\cM$ and $x^n = (x_1, \dots, x_n)$, as in Section \ref{S:defns}, then the first property is that the maximum likelihood estimator $\hat{\xi}$ of the expectation parameter $\xi$ is simply the mean, i.e.,
\beq
\label{E:hxixbar}
\hat{\xi}(x^n) = (x_1 + \dots + x_n)/n
\eeq
by Theorem 2.2.6 of \citep{KassVos97} and the comments preceding it.
The second property is that this estimator is exactly (as opposed to asymptotically) unbiased and efficient, meaning that
\beq
\label{E:meanvarxbar}
\E[\hat{\xi}] = \xi \mbox{ and } \var(\hat{\xi}) = (n g_\Xi(\xi))^{-1}
\eeq
by \citep[Theorems 2.2.1 and 2.2.5]{KassVos97}, where
$g_\Xi$ is the Fisher information metric on $\Xi$ and
the superscripted $-1$ indicates the matrix inverse.  The first property allows us to apply Edgeworth expansions directly to the maximum likelihood estimator $\hat{\xi}$ and the second property allows us to express the approximating distribution in terms of the Fisher information matrix.

Let $\phixi: \R^d \to \R$ be the probability density function (PDF) of the $d$-dimensional normal distribution $N_d(\xi,(n g_\Xi(\xi))^{-1})$ which has the same mean and variance as $\hat{\xi}$, so that
\beq
\label{E:phixi}
\phixi(x) = \left( \frac{n}{2\pi} \right)^{d/2} \sqrt{\det g_\Xi(\xi)} \exp\left( -\frac{n}{2} (x-\xi)^T g_\Xi(\xi) (x-\xi) \right)
\eeq
for any $x \in \R^d$.  As in Section \ref{S:defns} (but with $\xi$, $\hxi$, $\Xi$, $x$ in place of $\theta$, $\htheta$, $\Theta$, $t$), let the distribution of the maximum likelihood estimator $\hat{\xi}$ of $\xi$ be $q_n(\cdot,\xi) \nu_n$.
For any fixed $\xi \in \Xi$, (\ref{E:hxixbar}) and the central limit theorem imply that this distribution can be approximated as
\begin{align}
\label{E:CLT}
\int_B q_n(x,\xi) d\nu_n(x) = \int_B \phixi(x) d\lambda(x) + o(1)
\end{align}
as $n \to \infty$, where $\lambda$ is the Lebesgue measure on $\R^d$ and the approximation is uniform for all Borel sets $B \subseteq \Xi$ satisfying some weak regularity conditions (e.g. see \citep[Corollary 1.4]{BhattacharyaDenker90}).  Edgeworth expansions \citep{BarndorffNielsenCox79,BhattacharyaDenker90} refine this result to imply, for fixed $\xi$ and for any integer or half-integer $s \ge 0$, that there is a
function $H_s: \Xi \times \Xi \to \R$
so that
\begin{align}
\label{E:edgeworth_firstdefn}
\int_B q_n(x,\xi) d\nu_n(x) = \int_B H_s(x,\xi) \phixi(x) d\lambda(x) + o(n^{-s})
\end{align}
as $n \to \infty$, where the error is uniform for all Borel sets $B \subseteq \Xi$ satisfying weak regularity conditions, as for (\ref{E:CLT}).

The function $H_s(x,\xi)$ appearing in (\ref{E:edgeworth_firstdefn}) is a polynomial in $x$ and only depends on $\xi$ through the cumulants of the distribution in $\cM$ corresponding to $\xi$.
In Section \ref{S:integrand} we will give explicit formulae for the function $H_s(x,\xi)$ in the case of most interest to us.
However, even without an explicit formula, the defining property (\ref{E:edgeworth_firstdefn}) of $H_s(x,\xi)$ allows us to state and prove the following theorem, which is the main ingredient in our proof of (\ref{E:PC_fullapprox}).

\begin{theorem}
\label{T:PC}
Let $\cM$ and $K \subseteq \Theta$ satisfy the regularity conditions given at the end of Section \ref{S:defns}, where $\Theta$ is any parameter space for $\cM$.  Let $f:\Theta \to \Xi$ be the reparameterisation map from $\Theta$ to the expectation parameter space $\Xi$. Then for any integer $s\ge 0$,
\begin{align}
\label{E:PC_theta}
\comp(K) = \frac{d}{2}\log \frac{n}{2\pi} + \log \int_K H_s(f(\theta),f(\theta)) \; d\pi_\Theta(\theta) + o(n^{-s})
\end{align}
as $n \to \infty$.
\end{theorem}

\begin{proof}
See Section \ref{S:proof_T:PC} (and see Remark \ref{R:proofidea} for a heuristic description of the proof).
\end{proof}

\begin{remark}
Edgeworth expansions have a number of known weaknesses, such as potentially giving signed distributions (e.g., partly negative PDFs) and only controlling absolute errors, which are not very informative in the tails of distributions.  However, neither of these defects affects our results, since we only need Edgeworth expansions in small neighbourhoods of the mean.
\end{remark}

\section{The integrand on the natural parameter space}
\label{S:integrand}

In this section, we complete the proof of (\ref{E:PC_fullapprox}) by calculating the integrand $H_s(f(\theta),f(\theta))$ appearing in (\ref{E:PC_theta}) for the special case when $\Theta$ is the natural parameter space of the exponential family.  Note that the corresponding function for any other parameter space is simply obtained by composing the function found in this section with the appropriate reparameterisation map (see the first paragraph of Section \ref{S:proof_T:PC}).

So let $\Theta$ be the natural parameter space and let $\psi: \Theta \to \R$ be the log-partition function of the natural exponential family $\cM$, so that the distribution of one data point $x$ is
\begin{align}
\label{E:natparam}
e^{x \cdot \theta - \psi(\theta)} \mu_1
\end{align}
for some measure $\mu_1$ on $\R^d$ which does not depend on $\theta$.
The cumulant generating function of the distribution (\ref{E:natparam}) is $K(t) = \psi(t+\theta) - \psi(\theta)$
\citep[eq. 2.2.4]{KassVos97}, so for any integer $r \ge 1$ and any $i_1, \dots, i_r \in \{1, \dots, d\}$, the cumulant $\kappa_{i_1 \dots i_r}$ of the distribution (\ref{E:natparam}) is
\begin{align}
\label{E:cumulant}
\kappa_{i_1 \dots i_r}
= \left.\frac{\partial^r K}{\partial t_{i_1} \dots \partial t_{i_r}}\right|_{t=0}
= \left.\frac{\partial^r \psi}{\partial \theta_{i_1} \dots \partial \theta_{i_r}}\right|_{\theta=\theta}.
\end{align}
Let $g^{ij} = (g_\Theta(\theta)^{-1})_{ij}$ be the $(i,j)^{th}$ component of the matrix inverse of the Fisher information metric $g_\Theta(\theta)$ on $\Theta$.

\begin{theorem}
\label{T:Edgeworth_fn}  For $\Theta$ the natural parameter space, the integrand appearing in (\ref{E:PC_theta}) is given by
$$ H_s(f(\theta),f(\theta)) = \sum_{i=0}^s  F_i(\theta) n^{-i}$$
for any $\theta \in \Theta$, where each $F_i(\theta)$ is a function of $\theta$ but not $n$.  In particular, $F_0(\theta) = 1$,
\begin{align}
\label{E:S1}
F_1(\theta) =& \frac{1}{8} \sum_{i_1, \dots, i_4}  \kappa_{i_1 \dots i_4} g^{i_1i_2}g^{i_3i_4}  -
\frac{1}{8} \sum_{i_1, \dots, i_6}  \kappa_{i_1 i_2 i_3} \kappa_{i_4 i_5 i_6} g^{i_1i_2}g^{i_3i_4}g^{i_5i_6} \nonumber \\
& - \frac{1}{12} \sum_{i_1, \dots, i_6}  \kappa_{i_1 i_2 i_3} \kappa_{i_4 i_5 i_6} g^{i_1i_4}g^{i_2i_5}g^{i_3i_6}
\end{align}
and
\begin{align}
F_2(\theta) =& \frac{1}{720} \sum_{i_1, \dots, i_6}  \kappa_{i_1 \dots i_6} h_{i_1 \dots i_6}  +
\frac{1}{720} \sum_{i_1, \dots, i_8}  \kappa_{i_1 i_2 i_3} \kappa_{i_4 \dots i_8} h_{i_1 \dots i_8} + \nonumber \\
& \frac{1}{1152} \sum_{i_1, \dots, i_8}  \kappa_{i_1 \dots i_4} \kappa_{i_5 \dots i_8} h_{i_1 \dots i_8}
+ \frac{1}{1728} \sum_{i_1, \dots, i_{10}}  \kappa_{i_1 i_2 i_3} \kappa_{i_4 i_5 i_6} \kappa_{i_7 \dots i_{10}} h_{i_1 \dots i_{10}} + \nonumber  \\
& \frac{1}{31104} \sum_{i_1, \dots, i_{12}}  \kappa_{i_1 i_2 i_3} \kappa_{i_4 i_5 i_6} \kappa_{i_7 i_8 i_9} \kappa_{i_{10} i_{11} i_{12}} h_{i_1 \dots i_{12}}, \label{E:S2}
\end{align}
where $h_{i_1 \dots i_r}$ is defined below and each index $i_k$ ranges from $1$ to $d$ in each sum, e.g.
$$ \sum_{i_1, \dots, i_4} \mbox{ is shorthand for } \sum_{i_1=1}^d \sum_{i_2=1}^d \sum_{i_3=1}^d \sum_{i_4=1}^d.$$
\end{theorem}

\begin{proof}
The proof just rescales the Edgeworth expansions from \citet{BarndorffNielsenCox79} and specialises them to exponential families, see Section \ref{S:proof_T:Edgeworth_fn} for details.
\end{proof}

The expressions $h_{i_1 \dots i_r}$ appearing in (\ref{E:S2}), where $r=2k$ is an even positive integer, are the (un-normalized) Hermite numbers given by
\begin{align}
\label{E:hermite_defn}
h_{i_1 \dots i_r}
&= \frac{(-1)^k}{k!} \; \frac{\partial^r}{\partial z_{i_1} \dots \partial z_{i_r}} \left[\left(\sum_{a,b=1}^d g^{ab} z_a z_b/2 \right)^{k} \right],
\end{align}
where $z_1, \dots, z_d$ are dummy variables, see Theorem \ref{T:Hermite_numbers} in Appendix \ref{S:appendixB}.  The expression in square brackets is a degree $r$ polynomial in $z_1, \dots, z_d$ so $h_{i_1 \dots i_r}$ does not depend on $z_1, \dots, z_d$.  Up to sign, $h_{i_1 \dots i_r}$ is a certain sum of products of components $g^{ab}$ of $g_\Theta(\theta)^{-1}$, e.g. see (\ref{E:conversion_h4}) and (\ref{E:conversion_h6}).

\begin{remark}
The un-normalized Hermite numbers (which corresponding to a normal distribution with a general, rather than identity, variance-covariance matrix) also arise in Gaussian processes and quantum field theory \citep[eqn. 1 and 2]{Nelson73}.
\end{remark}

\begin{remark}
\label{R:geom}
The geometrical or statistical significance of the functions $F_1(\theta), F_2(\theta), \dots$ is not clear.  $F_1(\theta)$ does not appear to be any sort of contraction of the Riemann curvature tensor $R$ (for either the Levi-Civita connection or any of Amari's other $\alpha$-connections \citep[\S 2.3]{AmariNagaoka00}) because in the natural parameter space, $R$ only involves third derivatives of the log-partition function $\psi$ whereas $F_1(\theta)$ involves fourth derivatives.  Some of the terms in (\ref{E:S1}) have a superficial similarity to Efron's curvature (see \citep{Efron75} or \citep[eqn. 4.3.10]{KassVos97}), but Efron's formula measures the extrinsic curvature of curved exponential families and so is $0$ for exponential families.  In Section \ref{S:AC}, we give an expression for $F_1(\theta)$ that is valid in any parameter space, in terms of contractions of products of the Amari-Chentsov tensors.
\end{remark}

\section{Examples}
\label{S:examples}

It is possible to calculate the functions $F_i(\theta)$ appearing in (\ref{E:PC_fullapprox}) by using the formulae of Theorem \ref{T:Edgeworth_fn} and a formula-manipulation program like Maxima \citep{maxima5.31.2}.  This is especially easy when $d$ is small, though it can also be done for all $d$, as we now illustrate.

\subsection{Exponential data}

Let $\cM$ be the family of exponential distributions on $\R$, so each observation $x$ is governed by some distribution of the form (\ref{E:natparam}), where $\mu_1$ is the Lebesgue measure restricted to the positive reals, the natural parameter space $\Theta$ is the set of negative reals, $-\theta > 0$ is the rate parameter of the exponential distribution and the log-partition function is
$$\psi(\theta) = - \log(- \theta).$$
Each distribution in $\cM$ is continuous so it satisfies Cram\'er's condition by \citep[p. 57]{Hall92}.

Since the parameter space $\Theta$ of this family has dimension $d=1$, the indices $i_k$ in (\ref{E:S1}) and (\ref{E:S2}) all take the value $i_k=1$, so each sum consists of a single term.  Also, $g_{11}$ is the second derivative of $\psi(\theta)$ by (\ref{E:FI_defn}), so
$g_{11} = \theta^{-2}$, and the matrix inverse of the $1 \times 1$ matrix $g_\Theta$ has the single component $g^{11}
= \theta^2$.  Also, by (\ref{E:cumulant}), the cumulants appearing in Theorem \ref{T:Edgeworth_fn} are just the derivatives of $\psi(\theta)$, e.g. $\kappa_{111} = -2 \theta^{-3}$ and $\kappa_{1111} = 6 \theta^{-4}$.  So by (\ref{E:S1}),
\begin{align*}
F_1(\theta)
=& \frac{1}{8} \kappa_{1111} g^{11}g^{11} - \frac{1}{8} \kappa_{111} \kappa_{111} g^{11}g^{11}g^{11}
- \frac{1}{12} \kappa_{111} \kappa_{111} g^{11}g^{11}g^{11} \\
=& \frac{1}{8} (6 \theta^{-4}) (\theta^2)^2 - \frac{1}{8} (-2 \theta^{-3})^2 (\theta^2)^3 - \frac{1}{12} (-2 \theta^{-3})^2 (\theta^2)^3\\
=& -1/12.
\end{align*}
A similar calculation, based on (\ref{E:S2}) and aided by Maxima \citep{maxima5.31.2}, shows that
$$ F_2(\theta) = 1/288. $$
Since we have just shown that $F_1(\theta)$ and $F_2(\theta)$ are constant in $\theta$, we will now write $F_1$ and $F_2$ for these constant functions.  Then (\ref{E:PC_fullapprox}) implies
\begin{align}
\comp(K)
&= \frac{d}{2}\log \frac{n}{2\pi} + \log \int_K \left( 1 +  F_1 n^{-1} + F_2 n^{-2} \right) d\pi_\Theta(\theta) + O(n^{-3}) \nonumber \\
&= \frac{d}{2}\log \frac{n}{2\pi} + \log \int_K d\pi_\Theta(\theta) + \log \left( 1 +  F_1 n^{-1} + F_2 n^{-2} \right) + O(n^{-3}) \nonumber \\
&= \frac{d}{2}\log \frac{n}{2\pi} + \log \vol(K) + \left(F_1 n^{-1} + F_2 n^{-2} \right) - \frac{1}{2} \left(F_1 n^{-1} + F_2 n^{-2} \right)^2 + O(n^{-3}) \nonumber \\
&= \frac{d}{2}\log \frac{n}{2\pi} + \log \vol(K) + F_1 n^{-1} + \left(F_2 - F_1^2/2 \right) n^{-2} + O(n^{-3}),  \label{E:PC_const}
\end{align}
where the second last step uses the asymptotic expansion $\log(1+z) = z - z^2/2 + O(z^3)$ as $z \to 0$.
Substituting the above values of $F_1$ and $F_2$ into (\ref{E:PC_const}) therefore gives
$$ \comp(K) = \frac{d}{2}\log \frac{n}{2\pi} + \log \vol(K) - \frac{1}{12 n} + O(n^{-3}) $$
because $F_2 - F_1^2/2$ vanishes.

Note that (\ref{E:PC_const}) holds for general $d$ and even for non-constant $F_1(\theta)$ and $F_2(\theta)$ if we replace each $F_i$ in (\ref{E:PC_const}) by the expected value of $F_i(\theta)$, where $\theta$ is regarded as a random variable whose distribution is proportional to the Jeffreys prior restricted to $K$.

\subsection{Spherical normal data with unknown variance}
\label{S:spherical_unknown}

We now consider a model where each observation $y$ is distributed according to the $(d-1)$-dimensional spherical normal distribution $N_{d-1}(\beta, \sigma^2 I_{d-1})$, where $I_{d-1}$ is the $(d-1) \times (d-1)$ identity matrix and the parameters $\beta \in \R^{d-1}$ and $\sigma > 0$ are to be estimated.  One reason for studying this model is that an exact formula for $\comp(K)$ is known, so in Section \ref{S:finite_sample} we will verify the expression for $F_1(\theta)$ given here.

For this model, a sufficient statistic $x$ and the corresponding natural parameter $\theta$ are given by
\begin{align}
x = \left[ \begin{array}{cc} y \\ \| y \|^2  \end{array} \right]
\mbox{ and }
\theta = \frac{1}{\sigma^2} \left[ \begin{array}{cc} \beta \\ -1/2  \end{array} \right],
\end{align}
where $y = [y_{1} \; \dots \; y_{d-1}]^T$ and $\| y \|^2 = y_{1}^2 + \dots + y_{d-1}^2$.  Then $x$ is governed by a member of a natural exponential family of the form (\ref{E:natparam}), with log-partition function
$$ \psi(\theta) = \left(\frac{1-d}{2}\right) \log(-2\theta_d) - (\theta_1^2 + \dots + \theta_{d-1}^2)/4\theta_d. $$
Note that Cram\'er's condition holds for each distribution in this natural exponential family, by \citep[Lemma 1.5]{BhattacharyaDenker90}.  The (scaled) Fisher information metric $g_\Theta(\theta)$ is the Hessian of $\psi(\theta)$, by (\ref{E:FI_defn}) or \citep[Theorem 2.2.5]{KassVos97}, and the matrix inverse of $g_\Theta$ has $(i,j)^{th}$ component
\begin{align}
g^{ij} =
\left\{ \begin{array}{ll}
\theta_i \theta_j (d-1)/2 - 2 \theta_d &\mbox{ if $i=j \not= d$} \\
\theta_i \theta_j (d-1)/2 &\mbox{ otherwise.}
\end{array} \right.
\end{align}
Also, the cumulants $\kappa_{i_1 \dots i_r}$ can be obtained by differentiating $\psi$, as in (\ref{E:cumulant}).  Then a calculation based on (\ref{E:S1}) (which was aided by Maxima \citep{maxima5.31.2} and which considered all combinations of the two cases $i_k < d$ and $i_k=d$ for each index $i_k$) shows that the three sums appearing in (\ref{E:S1}) are constant in $\theta$ and are equal to the expressions in square brackets below:
\begin{align}
F_1(\theta) =& \frac{1}{8} \left[\frac{8d+4}{d-1}\right]  -
\frac{1}{8} \left[\frac{2d^2+4d+2}{d-1}\right] - \frac{1}{12} \left[\frac{6d+2}{d-1}\right] \nonumber \\
= & \frac{1 - 3d^2}{12(d-1)}. \label{E:spherical_F1}
\end{align}
So when $s=1$, (\ref{E:PC_fullapprox}) becomes
\begin{align}
\comp(K)
=& \frac{d}{2}\log \frac{n}{2\pi} + \log \int_K \left[ 1 + \frac{1 - 3d^2}{12(d-1)n} \right] \; d\pi_\Theta(\theta) + o(n^{-1}) \nonumber \\
=& \frac{d}{2}\log \frac{n}{2\pi} + \log \int_K d\pi_\Theta(\theta) + \log \left[ 1 + \frac{1 - 3d^2}{12(d-1)n} \right] + o(n^{-1}) \nonumber \\
=& \frac{d}{2}\log \frac{n}{2\pi} + \log \vol(K) + \frac{1 - 3d^2}{12(d-1)n} + o(n^{-1}) \label{E:comp_spherical}
\end{align}
as $n \to \infty$.  Note that this is just the usual approximation (\ref{E:PC_volapprox}) plus the correction term $(1 - 3d^2)/12(d-1)n$.  We will investigate the finite-sample performance of the two approximations (\ref{E:comp_spherical}) and (\ref{E:PC_volapprox}) in Section \ref{S:finite_sample}, below.

\subsection{Spherical normal data with known variance}

For spherical normal data with known variance, the log-partition function $\psi(\theta)$ is quadratic in $\theta$, so all of the cumulants $\kappa_{i_1\ldots i_r}$ vanish for $r \ge 3$, and hence the functions $F_i(\theta)$ also vanish for $i \ge 1$. So by (\ref{E:PC_fullapprox}), the standard formula (\ref{E:PC_volapprox}) is accurate to order $o(n^{-s})$ for all integers $s\ge 0$.  This is a positive check on our formulae, because it is well-known that (\ref{E:PC_volapprox}) is the exact formula for this model, as we now briefly demonstrate.

Working in the expectation parameter space, if $x_1, \dots, x_n \sim N_d(\xi,\sigma^2 I_d)$ for $\sigma > 0$ known and $\xi \in \R^d$ unknown, then $\hxi = (x_1 + \dots + x_n)/n$ and hence $\hxi \sim N_d(\xi,(\sigma^2/n) I_d)$.  Therefore $\nu_n = \lambda$ and
$$ q_n(x,\xi) = \left(\frac{n}{2 \pi \sigma^2}\right)^{d/2} \exp\left( - \frac{n}{2 \sigma^2} \| x - \xi \|^2 \right) $$
so $q_n(\xi,\xi) = (n/2 \pi \sigma^2)^{d/2}$.  Hence if $K \subseteq \Xi$ then the exact formula for $\comp(K)$ is
$$ \comp(K) = \log \int_K q_n(\xi,\xi) d\nu_n(\xi) = \log \int_K (n/2 \pi \sigma^2)^{d/2} d\lambda(\xi)
=  \frac{d}{2}\log \frac{n}{2\pi} + \log \vol(K) $$
since $\pi_\Xi = \lambda \sigma^{-d}$.

\section{Finite sample performance}
\label{S:finite_sample}

In this section we compare the finite-sample performance of the standard approximation (\ref{E:PC_volapprox}) and its first-order correction (i.e., (\ref{E:PC_fullapprox}) with $s=1$), in the case of the spherical normal model with unknown variance (see Section \ref{S:spherical_unknown}).

\citet{Rissanen00} studied the linear regression model with the parameterisation $\theta = (\beta,\tau)$, where $\tau = \sigma^2$ is the noise variance and $\beta$ is the regression coefficient (note that this is not the natural parameterisation of this exponential family).  Let $t = (\hat{\beta}, \hat{\tau})$ be the maximum likelihood estimate for $\theta$.  If the model has $N$ observations, $k$ covariates and design matrix $A$ (i.e., $A$ is the $N \times k$ matrix whose columns are the covariates of the model) then Rissanen showed that the maximized likelihood is
\begin{align}
\label{E:preexact}
q(t,t) =  \frac{\sqrt{\det A^T A}}{(\pi N)^{k/2} \Gamma(\frac{N-k}{2})} \left(\frac{N}{2e}\right)^{N/2} \hat{\tau}^{-1-k/2}
\end{align}
with respect to the Lebesgue measure $\lambda$, where $\Gamma$ is the gamma function.  It is not hard to show
that the unscaled Jeffrey's prior $\bar{\pi}_\Theta$ for the above parameterisation is given by
\begin{align}
\label{E:jeffreys}
d\bar{\pi}_\Theta(\theta) = \tau^{-1-k/2} \sqrt{(N/2) \det A^T A} \; d\lambda(\theta),
\end{align}
where `unscaled' means that $\bar{\pi}_\Theta$ is the volume density of the unscaled Fisher information metric, which is $n$ times the right-hand side of (\ref{E:FI_defn}).
Therefore
\begin{align}
\comp(K)
&= \log \int_K q(t,t) d\lambda(t) \mbox{ by (\ref{E:PC_defn})} \nonumber \\
&= \log \int_K q(\theta,\theta) d\lambda(\theta) \mbox{ since $t$ is a dummy variable} \nonumber \\
&= \log \int_K  \frac{(\pi N)^{-k/2}}{ \Gamma(\frac{N-k}{2})} \sqrt{\frac{2}{N}} \left(\frac{N}{2e}\right)^{N/2} d\bar{\pi}_\Theta(\theta) \mbox{ by (\ref{E:preexact}) and (\ref{E:jeffreys})} \nonumber \\
&=  \log \int_K d\bar{\pi}_\Theta(\theta) + \log \left(\frac{(\pi N)^{-k/2}}{ \Gamma(\frac{N-k}{2})} \sqrt{\frac{2}{N}} \left(\frac{N}{2e}\right)^{N/2} \right) \label{E:exact}
\end{align}
where the last step follows because the integrand is constant in $\theta$.  Note that this is not an approximation, $\comp(K)$ is given exactly by (\ref{E:exact}).

Now, take $A$ to be the $N \times k$ matrix $A = [I_k \cdots I_k]^T$, where $k=d-1$, $N=nk$ and $I_k$ is the $k \times k$ identity matrix.  Then the regression response variable $[y_1 \cdots y_n]^T$ is distributed in such a way that $y_1, \dots, y_n \sim N_{d-1}(\beta, \sigma^2 I_{d-1})$ are IID, as in Section \ref{S:spherical_unknown}.  Substituting these values into (\ref{E:exact}) therefore gives the following exact formula for the parametric complexity $\comp(K)$ of the model of Section \ref{S:spherical_unknown}:
\begin{align}
\comp(K)
&= \log \vol(K) - \frac{k}{2}\log(\pi k) - \frac{1}{2}\log \frac{k}{2} + \frac{nk}{2}\log \frac{nk}{2e} - \log \Gamma\left( \frac{k(n-1)}{2} \right), \label{E:comp_exact}
\end{align}
where we have retained the notation $k=d-1$ for convenience and where, as is standard throughout this paper, $\vol(K) = \int_K d\pi_\Theta(\theta) = n^{-d/2} \int_K d\bar{\pi}_\Theta(\theta)$ is the volume of $K$ with respect to the scaled Jeffreys prior $\pi_\Theta$.  We can use this exact formula in two ways:  to assess the accuracy of our asymptotic expansion and to check the veracity of our calculations.

Figure \ref{F:vs_n} assesses the accuracy of both the standard approximation (\ref{E:PC_volapprox}) to $\comp(K)$, in the case of spherical normal data with unknown variance, and its first-order correction (\ref{E:comp_spherical}).  Note that the summand $\log \vol(K)$ cancels when comparing these two approximations to the exact formula (\ref{E:comp_exact}), so the the amount of  over- or underestimation of $\comp(K)$ does not depend on $K$.  Figure \ref{F:vs_n} shows that both formulae overestimate $\comp(K)$, and hence NML codelength, for all model dimensions $d$ and all sample sizes $n \ge d$, with the amount of overestimation increasing with $d$ and decreasing with $n$ (note that results cannot be calculated for the under-determined models with $n<d$.)  The overestimation is negligible for the corrected formula (\ref{E:comp_spherical}) when $n \ge 20$, but it is significant for the standard formula (\ref{E:PC_volapprox}) for moderate $n$, e.g. it is greater than $0.05$ for $n=50$ and $d=11$.

Substituting an asymptotic expansion for the gamma function into (\ref{E:comp_exact}) gives
\begin{align*}
\comp(K)
&= \frac{d}{2}\log \frac{n}{2\pi} + \log \vol(K)
- n^{-1}\left[ \frac{3d^2-1}{12(d-1)} \right] - n^{-2}\left[ \frac{d(d+1)}{12(d-1)} \right] + O(n^{-3}).
\end{align*}
Comparing the above formula to (\ref{E:comp_spherical}), which was obtained directly from Theorems \ref{T:PC} and \ref{T:Edgeworth_fn}, shows that our expression (\ref{E:spherical_F1}) for $F_1(\theta)$
is correct.
The coefficient of $n^{-2}$ in the above formula also verifies our calculation (not shown here) of $F_2(\theta)$ for $d=2$.
This is therefore a positive check on the formulae of Theorems \ref{T:PC} and \ref{T:Edgeworth_fn}.

\begin{figure}
\centering
\includegraphics[width=8cm]{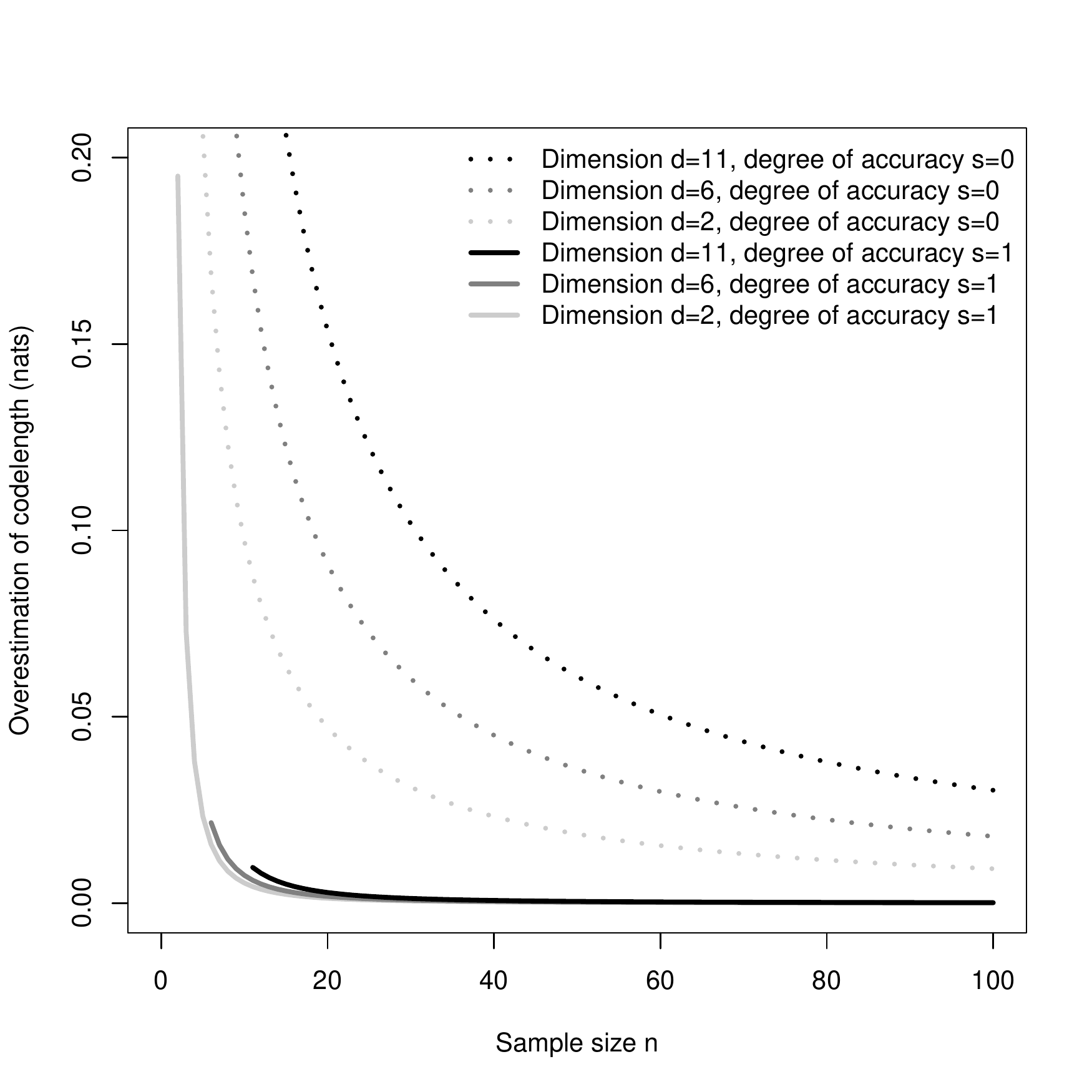} \\
\caption{The amount of overestimation of $\comp(K)$, and hence NML codelength, by both the standard approximation (\ref{E:PC_volapprox}), corresponding to $s=0$ in (\ref{E:PC_fullapprox}), and its first-order correction, corresponding to $s=1$ in (\ref{E:PC_fullapprox}), as a function of sample size $n$ and model dimension $d$, for the spherical normal model with unknown variance (see Section \ref{S:finite_sample}).
}
\label{F:vs_n}
\end{figure}

\section{The integrand on a general parameter space}
\label{S:AC}

In this section, we give an expression for the integrand $H_s(f(\theta),f(\theta))$ of (\ref{E:PC_theta}) which is valid in all parameter spaces.  The expression for $H_s(f(\theta),f(\theta))$ in Theorem \ref{T:Edgeworth_fn} is convenient for calculations, but it is not invariant under co-ordinate changes because cumulants are not tensors on $\Theta$.  The cumulant for a distribution is independent of the family to which it belongs (cumulants are derivatives of $K(t)$ with respect to the dummy variable $t$) so the cumulants $\kappa_{i_1 \dots i_r}$ transform like functions rather than the components of a tensor.  A co-ordinate-dependent expression like the one in Theorem \ref{T:Edgeworth_fn} is likely to be arbitrary in some ways, and therefore likely to hide the true mathematical structure, so in this section we give an expression for $F_1(\theta)$ which holds for all parameterisations.

For any $r=1, 2, 3, \dots$, define the $r^{th}$ Amari-Chentsov tensor $\ac{r}$ to be the symmetric tensor on $\Theta$ with components
\begin{align}
\label{E:ACdefn}
\ac{r}_{i_1, \dots, i_r} = \int \left( \prod_{j=1}^r \frac{\partial \log p_1(x, \theta) }{\partial \theta_{i_j}}  \right)  p_1(x, \theta) \; d\mu_1(x)
\end{align}
where $p_1(x, \theta)$ is the likelihood function at parameter $\theta$ for one observation $x$.  This is a tensor on any parameter space because it has a co-ordinate-independent definition (essentially just by replacing the partial differentials with arbitrary first-order partial differential operators, which are interpreted as vector fields in a way which is standard in differential geometry).

It is easy to show that the first Amari-Chentsov tensor $\ac{1}$ is trivial, $\ac{1} = 0$, but the second one is the Fisher information metric, $\ac{2} = g_\Theta$, so the Amari-Chentsov tensors can be seen as generalisations of the Fisher information metric tensor.  The third Amari-Chentsov tensor $\ac{3}$ is also interesting, because the difference between any two of Amari's $\alpha$-connections is a multiple of $\ac{3}$ \citep[\S 2.3]{AmariNagaoka00}.

In the natural parameterisation of an exponential family, the Amari-Chentsov tensors can be calculated using the following recurrence relation.

\begin{lemma}
\label{L:sk}
If $\Theta$ is the natural parameterisation of the exponential family and $r \ge 3$ is an integer then
$$\ac{r}_{i_1, \dots, i_r} = \frac{\partial }{\partial \theta_{i_r}} \ac{r-1}_{i_1, \dots, i_{r-1}}
+ \sum_{j=1}^{r-1} g_{i_j i_r} \ac{r-2}_{i_1, \dots,\hat{i}_j,\dots, i_{r-1}}$$
where the hat indicates that the term $i_j$ should be left out.
\end{lemma}

\begin{proof}
See Section \ref{S:proof_L:sk}.
\end{proof}

So in the natural parameterisation, using $\ac{1}_i=0$, $\ac{2}_{ij} = g_{ij} = \partial^2 \psi/\partial \theta_i\partial\theta_j$ and Lemma \ref{L:sk}, we have
\begin{align}
\label{E:AC3}
\ac{3}_{i_1 i_2 i_3} = \frac{\partial^3 \psi}{\partial \theta_{i_1} \partial\theta_{i_2} \partial\theta_{i_3}},
\end{align}
so another application of Lemma \ref{L:sk} implies
\begin{align}
\label{E:AC4}
\ac{4}_{i_1 \dots i_4} = \frac{\partial^4 \psi}{\partial \theta_{i_1} \dots \partial \theta_{i_4}}
+ g_{i_1 i_2}g_{i_3 i_4} + g_{i_1 i_3}g_{i_2 i_4} + g_{i_1 i_4}g_{i_2 i_3}.
\end{align}
The above derivatives of $\psi$ can be identified with the cumulants, by (\ref{E:cumulant}), so combining (\ref{E:AC3}) and (\ref{E:AC4}) with the expression for $F_1(\theta)$ in (\ref{E:S1}) gives
\begin{align}
\label{E:S1_invar}
F_1(\theta)
=& \frac{1}{8} \sum_{i_1, \dots, i_4}  \ac{4}_{i_1 \dots i_4} g^{i_1i_2}g^{i_3i_4}
 - \frac{1}{8} \sum_{i_1, \dots, i_6}  \ac{3}_{i_1 i_2 i_3} \ac{3}_{i_4 i_5 i_6} g^{i_1i_2}g^{i_3i_4}g^{i_5i_6} \nonumber  \\
& -\frac{1}{12} \sum_{i_1, \dots, i_6}  \ac{3}_{i_1 i_2 i_3} \ac{3}_{i_4 i_5 i_6} g^{i_1i_4}g^{i_2i_5}g^{i_3i_6} - \frac{d^2 + 2d}{8}
\end{align}
in the natural parameterisation, where the last term arises from contractions between the metric and its inverse, for example
$$\sum_{i_1, \dots, i_4} g_{i_1 i_2}g_{i_3 i_4} g^{i_1i_2}g^{i_3i_4}= d^2.$$
But since the left- and right-hand sides of (\ref{E:S1_invar}) are equivariant (i.e., transform correctly) under reparameterisations, and since we have just shown that (\ref{E:S1_invar}) is valid for the natural parameterisation, (\ref{E:S1_invar}) must hold for all parameterisations.

Using similar methods, an expression can also be given for $F_2(\theta)$, and higher terms, which is invariant under co-ordinate changes on the parameter space.

\section{Summary}
\label{S:conclusion}

We have given an asymptotic expansion for the contribution $\comp(K)$ to the parametric complexity coming from a given set $K$ in the parameter space, in the case of exponential families and IID data.  Our methods allow this expansion to be calculated to any desired degree of accuracy, and we gave explicit formulae for the first three terms.  The first term in this expansion is Rissanen's $o(1)$ approximation (\ref{E:PC_volapprox}), so our proof specialises to the first demonstration that Rissanen's result is actually $O(n^{-1})$.
Our proof is also arguably simpler than Rissanen's, since we do not have to partition $K$ into optimally sized pieces and since the theory of Edgeworth expansions largely takes care of effects at the boundary of $K$ for us.  We calculated the first-order correction to Rissanen's approximation for some examples, and showed that the resulting expression has improved finite-sample behaviour.  We also discussed the meaning of the higher-order terms, and we gave an expression for them which is valid in all parameterisations, in addition to giving a simpler expression in terms of cumulants which is valid for the natural parameterisation.

\appendix

\section{Proofs}
\label{S:proofs}

\subsection{Proof of Theorem \ref{T:PC}}
\label{S:proof_T:PC}

Under the reparameterisation map $f: \Theta \to \Xi$, the measure $\pi_\Theta$ corresponds to $\pi_\Xi$ and the function $\theta \mapsto H_s(f(\theta),f(\theta))$ corresponds to $\xi \mapsto H_s(\xi,\xi)$, so
\begin{align*}
\int_K H_s(f(\theta),f(\theta)) \; d\pi_\Theta(\theta) = \int_{f(K)} H_s(\xi,\xi) \; d\pi_\Xi(\xi).
\end{align*}
Also, if we interpret $K$ as a subset of $\Theta$ and $f(K)$ as a subset of $\Xi$, then
\begin{align*}
\comp(K) = \comp(f(K))
\end{align*}
by the definition (\ref{E:pc_defn_K}), since $\hxi = f \circ \htheta$ so the set of data points $x^n$ with $\htheta(x^n) \in K$ is the same as the set of $x^n$ with $\hxi(x^n) \in f(K)$.
Therefore the general case of the theorem will follow from the special case where $\Theta = \Xi$ and $f$ is the identity map,
so we restrict to this case now.

The Riemannian metric $g_\Xi$ can be used in a standard way to define a distance $d_F(x,\xi)$ between any two points $x,\xi \in \Xi$, so that $d_F(x,\xi)$ is equal to the infimum of the lengths of all smooth paths joining $x$ and $\xi$.  Then given any $x \in \Xi$ and any $\delta > 0$, let
$$B(x,\delta) = \{ \xi \in \Xi \mid d_F(x,\xi) \le \delta \}$$
be the closed ball in $\Xi$ of radius $\delta$ centred at $x$.  We assume that $\delta$ is small enough that $B(x,\delta)$ is compact for every $x \in K$  (i.e., that $\delta$ is less than the distance from $K$ to the `boundary' of $\Xi$).

Let $v_\delta$ be the volume of a Euclidean ball of radius $\delta$.  We will prove the theorem in the following three steps:
\begin{align*}
\exp \left( \comp(K) \right)
&= \int_K \int_{B(x,\delta)} v_\delta^{-1} q_n(x,\xi) \; d\pi_\Xi(\xi) d\nu_n(x) + O(\delta) \mbox{ Step 1}\\
&= \int_K \int_{B(x,\delta)} v_\delta^{-1} H_s(x,\xi) \phixi(x) d\pi_\Xi(\xi) d\lambda(x)
+ O(\delta) + o(n^{-s}) \mbox{ Step 2}\\
&= \left(\frac{n}{2\pi}\right)^{d/2} \int_K H_s(\xi,\xi) d\pi_\Xi(\xi) + O(\delta) + o(n^{-s}) \mbox{ Step 3}
\end{align*}
as $n \to \infty$ and $\delta \to 0$.  So taking the logarithm of both sides and setting $\delta = o(n^{-s})$ proves the theorem, since we will show that the error term $o(n^{-s})$ arising from Step 2 is uniform in $\delta$.

Steps 1 and 3 are similar, and are proved below by approximating the integrand to zeroth order in $\delta$ then integrating out the inner integral in the double integral (i.e., integrating in the $\xi$ direction).  Step 2 is obtained by swapping the order of integration then using Edgeworth expansions in the inner integral (since $\xi$ is fixed there).

{\em Step 1.}  Let
$$R = \{(x,\xi)\in \Xi \times \Xi \mid x \in K \mbox{ and } \xi \in B(x,\delta) \}$$
be the region of integration in the double integrals above.  Then for any $(x,\xi)$ in $R$, the zeroth order Taylor series expansion of $\xi \mapsto q_n(x,\xi)$ about $x$ shows that
\begin{align}
\label{E:lik_zeroth}
q_n(x,\xi) = q_n(x,x) + O(\delta),
\end{align}
where the error term $O(\delta)$ is uniformly bounded for all $(x,\xi) \in R$ because $q_n(x,\xi)$ is a smooth function and $R$ is compact.  Also, the Fisher information volume of the ball $B(x,\delta)$ can be approximated by the volume $v_\delta$ of a Euclidean ball of the same radius
\begin{align}
\label{E:vol_zeroth}
\int_{B(x,\delta)} d\pi_\Xi(\xi) = v_\delta + O(\delta)
\end{align}
by \citet{Gray73}, where the error term is uniformly bounded for all $x \in K$.  Therefore
\begin{align*}
& \int_K \int_{B(x,\delta)} v_\delta^{-1} q_n(x,\xi) \; d\pi_\Xi(\xi) d\nu_n(x)  \\
&= \int_K \int_{B(x,\delta)} v_\delta^{-1} q_n(x,x) \; d\pi_\Xi(\xi) d\nu_n(x) + O(\delta) \mbox{ by (\ref{E:lik_zeroth})} \\
&= \int_K v_\delta^{-1} q_n(x,x) \left( \int_{B(x,\delta)} d\pi_\Xi(\xi) \right) d\nu_n(x) + O(\delta) \\
&= \int_K q_n(x,x) d\nu_n(x) + O(\delta) \mbox{ by (\ref{E:vol_zeroth})}\\
&= \exp(\comp(K)) + O(\delta) \mbox{ by (\ref{E:PC_defn})}
\end{align*}
which proves Step 1.

{\em Step 2.}  Define $C(\xi,\delta)$ to be the truncated ball $K \cap B(\xi,\delta)$
and define $N(K,\delta)$
to be the closed $\delta$-neighbourhood
of $K$ in $\Xi$.
Then
\begin{align}
\label{E:deltaK_horiz}
x \in K \mbox{ and } \xi \in B(x,\delta) \Leftrightarrow \xi \in N(K,\delta) \mbox{ and } x \in C(\xi,\delta)
\end{align}
so the region of integration $R$ consists of all pairs $(x,\xi) \in \Xi \times \Xi$ satisfying either of these conditions.
Therefore
\begin{align*}
& \int_K \int_{B(x,\delta)} v_\delta^{-1} q_n(x,\xi) \; d\pi_\Xi(\xi) d\nu_n(x)\\
&=  \int_{N(K,\delta)} \int_{C(\xi,\delta)} v_\delta^{-1} q_n(x,\xi) \; d\nu_n(x) d\pi_\Xi(\xi) \mbox{ by (\ref{E:deltaK_horiz}) and Fubini's theorem} \\
&=  \int_{N(K,\delta)} \int_{C(\xi,\delta)} v_\delta^{-1} H_s(x,\xi) \phixi(x) \; d\lambda(x) d\pi_\Xi(\xi) + o(n^{-s}) \mbox{ by (\ref{E:edgeworth_firstdefn})} \\
&=  \int_K \int_{B(x,\delta)} v_\delta^{-1} H_s(x,\xi) \phixi(x) \; d\pi_\Xi(\xi) d\lambda(x)  + o(n^{-s}) \mbox{ by (\ref{E:deltaK_horiz}) and Fubini's theorem}
\end{align*}
so Step 2 is proved.  The fact that the error term $o(n^{-s})$ above is uniform in $\delta$ follows from \citep[Corollary 1.4]{BhattacharyaDenker90} with $\mathcal{A} = \{ K \cap B(\xi,\delta) \mid \xi \in \Xi \mbox{ and } \delta > 0 \}$, where $K$ is fixed.

{\em Step 3.}  For any $(x,\xi)$ in the region of integration $R$, the zeroth order Taylor series expansion of $\xi \mapsto H_s(x,\xi) \phixi(x)$ about $x$ shows that
\begin{align}
H_s(x,\xi)\phixi(x)
&= H_s(x,x)\phix(x)  + O(\delta) \nonumber \\
&= \left( \frac{n}{2\pi} \right)^{d/2} H_s(x,x) \sqrt{\det g_\Xi(x)} + O(\delta), \label{E:phixi_zeroth}
\end{align}
where the second equality follows by the definition (\ref{E:phixi}) of $\phixi$.  Here the error term $O(\delta)$ is uniformly bounded for $(x,\xi) \in R$, since $H_s(x,\xi) \phixi(x)$ is a smooth function and $R$ is compact.
Therefore
\begin{align*}
& \int_K \int_{B(x,\delta)} v_\delta^{-1} H_s(x,\xi) \phixi(x) \; d\pi_\Xi(\xi) d\lambda(x) \\
&= \left( \frac{n}{2\pi} \right)^{d/2} \int_K \int_{B(x,\delta)} v_\delta^{-1} H_s(x,x)\sqrt{\det g_\Xi(x)}   \; d\pi_\Xi(\xi) d\lambda(x) + O(\delta) \mbox{ by (\ref{E:phixi_zeroth})} \\
&= \left( \frac{n}{2\pi} \right)^{d/2}\int_K v_\delta^{-1} H_s(x,x)\sqrt{\det g_\Xi(x)} \left( \int_{B(x,\delta)} d\pi_\Xi(\xi) \right) d\lambda(x) + O(\delta) \\
&= \left( \frac{n}{2\pi} \right)^{d/2} \int_K H_s(x,x)\sqrt{\det g_\Xi(x)} d\lambda(x) + O(\delta) \mbox{ by (\ref{E:vol_zeroth})} \\
&= \left( \frac{n}{2\pi} \right)^{d/2} \int_K H_s(x,x)\; d\pi_\Xi(x) + O(\delta) \\
&= \left( \frac{n}{2\pi} \right)^{d/2} \int_K H_s(\xi,\xi)\; d\pi_\Xi(\xi) + O(\delta)
\end{align*}
since $\pi_\Xi = \lambda \sqrt{\det g_\Xi}$ by definition of the Jeffreys prior and since $x$ and $\xi$ are dummy variables in the last two integrals.  This proves Step 3 and hence the theorem.

\subsection{Proof of Theorem \ref{T:Edgeworth_fn}}
\label{S:proof_T:Edgeworth_fn}

Recall that $x_1, \dots, x_n$ are IID random variables governed by an unknown element of the exponential family $\cM$.  Suppose that the data-generating distribution has expectation parameter $\xi \in \Xi$, and consider this to be fixed throughout this section.  The maximum likelihood estimate $\hxi$ of $\xi$ is simply the sample mean $(x_1 + \dots + x_n)/n$, by (\ref{E:hxixbar}), and its mean and variance are $\E[\hat{\xi}] = \xi$ and $\var(\hat{\xi}) = (n g_\Xi(\xi))^{-1}$, by (\ref{E:meanvarxbar}).  Then the Edgeworth expansions from \citet{BarndorffNielsenCox79} or \citet[\S 4.5]{KassVos97} approximate the distribution $\zeta_n$ of the transformed mean $Z_n = \sqrt{n}(\hxi - \xi)$ by the asymptotic expansion
\begin{align}
\label{E:edgeworth_y}
\int_B d \zeta_n(z) = \int_B \widetilde{H}_s(z) \phi(z) d\lambda(z) + o(n^{-s})
\end{align}
as $n \to \infty$, where $\widetilde{H}_s(z)$ is an explicit polynomial in $z$, $\phi$ is the PDF for $N_d(0,g_\Xi(\xi)^{-1})$ and $B$ is any Borel set satisfying some mild regularity conditions.

To find $H_s(x,\xi)$ in terms of $\widetilde{H}_s(z)$, we simply have to calculate the effect of (the inverse of) the transformation $\tau: x \mapsto \sqrt{n}(x - \xi)$ on (\ref{E:edgeworth_y}).  The transformation $\tau$ takes $\hxi$ to $Z_n$, so $\tau^{-1}$ takes the distribution in the left-hand side of (\ref{E:edgeworth_y}) to the distribution $d q_n(x,\xi) d\nu_n(x)$ of $\hat{\xi}$.  On the other hand, the change of variables formula for PDFs shows that
if $W$ is any continuous random variable and $\rho$ is the PDF of $\tau(W)$ then the PDF for $W$ is $J_\tau \; \rho \circ \tau$,
where $J_\tau = n^{d/2}$ is the Jacobian determinant of $\tau$.  Therefore $\tau^{-1}$ takes the distribution in the right-hand side of (\ref{E:edgeworth_y}) to $\widetilde{H}_s(\sqrt{n}(x - \xi)) \phixi(x) d\lambda(x)$, so (\ref{E:edgeworth_y}) transforms under $\tau^{-1}$ to
$$ \int_B d q_n(x,\xi) d\nu_n(x) = \int_B \widetilde{H}_s(\sqrt{n}(x - \xi)) \phixi(x) d\lambda(z) + o(n^{-s}). $$
Comparing this equation to (\ref{E:edgeworth_firstdefn}) gives $H_s(x,\xi) = \widetilde{H}_s(\sqrt{n}(x- \xi))$, so
\begin{align}
\label{E:conversion_0}
H_s(\xi,\xi) = \widetilde{H}_s(0).
\end{align}

Now, the function $\widetilde{H}_s(z)$ of \citep{BarndorffNielsenCox79} or \citep[\S 4.5]{KassVos97} is a sum of terms involving integral and half-integral powers of $n$, but all of the terms involving half-integral powers vanish at $z=0$ by the theorem in the Appendix of \citep{BarndorffNielsenCox79}.  So (\ref{E:conversion_0}) implies that
$$H_s(f(\theta),f(\theta)) = \sum_{i=0}^s  F_i(\theta) n^{-i}$$
for integral $s$, as claimed.  The explicit formulae of \citep[\S 4.5]{KassVos97} then show that $F_0(\theta)=1$ and
\begin{align}
\label{E:conversion_S1}
F_1(\theta) = \frac{1}{24} \sum_{i_1, \dots, i_4}  \kappa_{i_1 \dots i_4} h_{i_1 \dots i_4}  +
\frac{1}{72} \sum_{i_1, \dots, i_6}  \kappa_{i_1 i_2 i_3} \kappa_{i_4 i_5 i_6} h_{i_1 \dots i_6}.
\end{align}
The methods of \citep[\S 4.5]{KassVos97} allow any $F_i(\theta)$ to be calculated, for example see Appendix \ref{S:appendixB}, where the formula (\ref{E:S2}) for $F_2(\theta)$ is calculated explicitly.  Therefore the theorem will be proved if we can convert (\ref{E:conversion_S1}) to the formula for $F_1(\theta)$ which appears in the statement of the theorem.

From \citep[\S 4.5]{KassVos97} or directly from the definition (\ref{E:hermite_defn}), the first few Hermite numbers are $h_{i_1 i_2} = -g^{i_1i_2}$,
\begin{align}
\label{E:conversion_h4}
h_{i_1 \dots i_4} = g^{i_1i_2}g^{i_3i_4} + g^{i_1i_3}g^{i_2i_4} + g^{i_1i_4}g^{i_2i_3}
\end{align}
and
\begin{align}
\label{E:conversion_h6}
h_{i_1 \dots i_6} = -(&
g^{i_1i_2}g^{i_3i_4}g^{i_5i_6} + g^{i_1i_3}g^{i_2i_4}g^{i_5i_6} + g^{i_1i_4}g^{i_2i_3}g^{i_5i_6} + g^{i_1i_2}g^{i_3i_5}g^{i_4i_6} + \nonumber \\
&g^{i_1i_3}g^{i_2i_5}g^{i_4i_6} + g^{i_1i_5}g^{i_2i_3}g^{i_4i_6} + g^{i_1i_2}g^{i_3i_6}g^{i_4i_5} + g^{i_1i_3}g^{i_2i_6}g^{i_4i_5} + \nonumber \\
&g^{i_1i_6}g^{i_2i_3}g^{i_4i_5} + g^{i_1i_4}g^{i_2i_5}g^{i_3i_6} + g^{i_1i_5}g^{i_2i_4}g^{i_3i_6} + g^{i_1i_4}g^{i_2i_6}g^{i_3i_5} + \nonumber \\
&g^{i_1i_6}g^{i_2i_4}g^{i_3i_5} + g^{i_1i_5}g^{i_2i_6}g^{i_3i_4} + g^{i_1i_6}g^{i_2i_5}g^{i_3i_4} ).
\end{align}
Substituting these into (\ref{E:conversion_S1}) and using the fact that the cumulants $\kappa_{i_1 \dots i_r}$ are invariant under all permutations of $i_1, \dots, i_r$, by (\ref{E:cumulant}), then
completes the proof.

\subsection{Proof of Lemma \ref{L:sk}}
\label{S:proof_L:sk}

Let $\Theta$ be the natural parameter space, so $\log p_1(x, \theta) = x \cdot \theta - \psi(\theta)$, and let $\partial_i = \partial/\partial\theta_i$.  Then
\begin{align}
\label{E:gexpfam}
g_{ij} = \partial_i \partial_j \psi  = -\partial_i \partial_j \log p_1,
\end{align}
by \citep[Theorem 2.2.5]{KassVos97}.  In particular, this implies that the right-hand side of (\ref{E:gexpfam}) is constant in $x$.  Then
\begin{align*}
\ac{r}_{i_1, \dots, i_r}
&= \int \left( \prod_{j=1}^r \partial_{i_j} \log p_1   \right)  p_1 \; d\mu_1(x) \mbox{ by (\ref{E:ACdefn})} \\
&= \int \left( \prod_{j=1}^{r-1} \partial_{i_j} \log p_1   \right)  \partial_{i_r} p_1 \; d\mu_1(x)\\
&= \partial_{i_r} \int \left( \prod_{j=1}^{r-1} \partial_{i_j} \log p_1   \right)  p_1 \; d\mu_1(x)
- \int \partial_{i_r} \left( \prod_{j=1}^{r-1} \partial_{i_j} \log p_1   \right)  p_1 \; d\mu_1(x) \\
&= \partial_{i_r} \ac{r-1}_{i_1, \dots, i_{r-1}}
- \int  \sum_{k=1}^{r-1} \left(  (\partial_{i_r} \partial_{i_k} \log p_1) \prod_{j\not=k} \partial_{i_j} \log p_1   \right)  p_1 \; d\mu_1(x) \\
&= \partial_{i_r} \ac{r-1}_{i_1, \dots, i_{r-1}}
+ \sum_{k=1}^{r-1} g_{i_r i_k} \int  \left( \prod_{j\not=k} \partial_{i_j} \log p_1   \right)  p_1 \; d\mu_1(x)
 \mbox{ by (\ref{E:gexpfam})} \\
&= \partial_{i_r} \ac{r-1}_{i_1, \dots, i_{r-1}}
+ \sum_{k=1}^{r-1} g_{i_r i_k} \ac{r-2}_{i_1, \dots,\hat{i}_k,\dots, i_{r-1}}  \mbox{ by (\ref{E:ACdefn}).}
\end{align*}

\section{Edgeworth expansions and Hermite numbers}
\label{S:appendixB}

This section gives a fairly self-contained description of multivariate Edgeworth expansions and the Hermite numbers for random variables with general (rather than identity) variance-covariance matrices.  We use the rigorous error analysis of \citet{BhattacharyaDenker90} and we closely follow the approach of \citet[\S4.5]{KassVos97}, except that we use multi-index notation which is more convenient when describing objects specified by large numbers of indices (such as high-degree partial derivatives of a function).  This section is an exposition of known results only, though its slight reformulation of Edgeworth expansions and its explicit formulae for higher order terms are useful for the rest of this paper, and do not appear in standard references.  Our main results are also reformulated in standard tensor notation at the end of this section.

Let $Y$ and $Y_1, Y_2, Y_3, \dots$ be IID $d$-dimensional random variables with mean $0$ and variance-covariance matrix $\Sigma$.  We assume that the distribution of $Y$ belongs to an exponential family of distributions, and therefore has a moment generating function, but we consider the distribution of $Y$ to be fixed for this entire section.  We aim to calculate the Edgeworth expansion for the distribution $\zeta_n$ of $Z_n \defeq (Y_1 + \dots + Y_n)/\sqrt{n}$ as $n \to \infty$ in terms of the cumulants of $Y$.

Let $M_Y$ and $K_Y$ be (respectively) the moment and cumulant generating functions of $Y$, given by
$ M_Y(t) \defeq \E [e^{t\cdot Y}]$ and $K_Y(t) = \log M_Y(t)$ for any $t \in \R^d$, where the dot represents the Euclidean inner product on $\R^d$.  Any $d$-tuple $\alpha = (\alpha_1, \dots, \alpha_d)$ of non-negative integers is called a multi-index, and we define $|\alpha| \defeq \alpha_1 + \dots + \alpha_d$,
$$\partial_t^\alpha = \frac{\partial^{|\alpha|}}{\partial t_1^{\alpha_1} \dots \partial t_d^{\alpha_d}},$$
$\alpha ! = \alpha_1 ! \dots \alpha_d !$ and $t^\alpha = t_1^{\alpha_1} \dots t_d^{\alpha_d}$ for any $t \in \R^d$.
Then the $\alpha^{th}$ cumulant $\kappa_\alpha$ of $Y$ is defined to be
\begin{align}
\label{E:defnkappa}
\kappa_\alpha = \left.(\partial_t^\alpha K_Y)\right|_{t=0}.
\end{align}
If $e_1, \dots, e_d$ is the standard basis for $\R^d$ then each $e_i$ is a multi-index, and a direct calculation from the definition shows that $\kappa_{e_i} = \E[Y^{(i)}]$ where $Y^{(i)}$ is the $i^{th}$ component of $Y = (Y^{(1)}, \dots, Y^{(d)})$.  We also note for future reference that $\kappa_{e_i+e_j} = \Sigma_{ij}$ (directly from the definition), so the degree-$2$ cumulants of $Y$ essentially give the covariance matrix of $Y$.  Then by (\ref{E:defnkappa}), the Taylor series expansion
for $K_Y$ about $0$ can be written as
\begin{equation}\label{E:CGF_Taylor}
K_Y(t) = \sum_{|\alpha| \ge 2}  \kappa_\alpha t^\alpha / \alpha!
\end{equation}
where the sum is over all multi-indices $\alpha$ with $|\alpha| \ge 2$, since $\kappa_\alpha=0$ when $|\alpha|$ is $0$ or $1$, because $K_Y(0)=0$ always and $\E [Y]=0$ by assumption.

Let $\phi:\R^d \to \R$ be the PDF for the $d$-dimensional normal distribution $N_d(0,\Sigma)$ with the same mean and covariance matrix as $Y$, so
\begin{equation}\label{E:PDF_normal}
\phi(y) = (2\pi)^{-d/2} (\det \Sigma)^{-1/2} \exp(-y^T \Sigma^{-1} y/2)
\end{equation}
for any $y \in \R^d$, where $y^T$ is the transpose of $y$ (here thinking of $y$ as a column matrix).
Then define the Hermite polynomial $h_\alpha(y)$ corresponding to multi-index $\alpha$ by
\begin{equation}\label{E:defn_Hermite}
h_\alpha(y) \phi(y) = (-1)^{|\alpha|} \; \partial_y^\alpha \phi(y)
\end{equation}
where $\partial_y^\alpha$ is a partial differential operator with respect to $y$ which is defined similarly to $\partial_t^\alpha$.  Note that many authors take $\Sigma$ to be the identity when defining the Hermite polynomials, but $\Sigma$ will be general throughout this section.

Let $\mathcal{A}$ be a class of Borel subsets of $\R^d$ which satisfies the condition in \citep[Corollary 1.4]{BhattacharyaDenker90}, e.g. $\mathcal{A}$ could consist of all convex Borel sets \citep[Remark 1.4.1]{BhattacharyaDenker90}.  Let $\lambda$ be the Lebesgue measure on $\R^d$.

\begin{theorem}
\label{T:edgeworth}
Suppose the distribution of $Y$ satisfies Cram\'er's condition (\ref{E:condCramer}).  Then for any integral of half-integral $s \ge 0$, there exist polynomials $S_{i/2}(y)$ in $y$ for each $i = 0, \dots, 2s$ so that
$$ \int_B d\zeta_n(y) = \int_B \phi(y) \left[ \sum_{i=0}^{2s} S_{i/2}(y) n^{-i/2} \right]d\lambda(y) + o(n^{-s})$$
uniformly in all Borel sets $B \in \mathcal{A}$.  In particular,  $S_0(y) = 1$,
$$ S_1(y) = \sum_{|\alpha| = 4}  \frac{\kappa_\alpha h_\alpha(y)}{\alpha!}  +
\sum_{|\alpha| = |\beta| = 3}  \frac{\kappa_\alpha \kappa_\beta h_{\alpha + \beta}(y)}{2 \; \alpha! \beta!} $$
and
\begin{align*}
S_2(y) &= \sum_{|\alpha| = 6}  \frac{\kappa_\alpha h_\alpha(y)}{\alpha!}  +
\sum_{\substack{|\alpha|=3\\|\beta| = 5}}  \frac{\kappa_\alpha \kappa_\beta h_{\alpha + \beta}(y)}{\alpha! \beta!} +
\sum_{|\alpha| = |\beta| = 4}  \frac{\kappa_\alpha \kappa_\beta h_{\alpha + \beta}(y)}{2 \; \alpha! \beta!} \\
& + \sum_{\substack{|\alpha| = |\beta| = 3\\ |\gamma|=4}}  \frac{\kappa_\alpha \kappa_\beta \kappa_\gamma h_{\alpha + \beta+ \gamma}(y)}{2 \; \alpha! \beta! \gamma!} +
\sum_{\substack{|\alpha| = |\beta| = 3 \\ |\gamma|=|\delta|=3}}  \frac{\kappa_\alpha \kappa_\beta \kappa_\gamma \kappa_\delta h_{\alpha + \beta+ \gamma + \delta}(y)}{24 \; \alpha! \beta! \gamma! \delta!}.
\end{align*}
Also, if $i$ is odd then $S_{i/2}(y)$ is a linear combination of terms of the form $h_\alpha(y)$ where $|\alpha|$ is odd (which implies $S_{i/2}(y)=0$ in our main case of interest, $y=0$, as we will see in Theorem \ref{T:Hermite_numbers}, below).
\end{theorem}

\begin{proof}
As mentioned above, this proof closely follows the approach of \citep[\S4.5]{KassVos97} and relies on the rigorous error analysis of \citep{BhattacharyaDenker90}.

We first note that
\begin{align*}
K_{Z_n}(t)
&= n K_Y(t/\sqrt{n}) \mbox{ since $Y_1, \dots, Y_n$ are IID} \\
&= \sum_{|\alpha| \ge 2}  \kappa_\alpha t^\alpha n^{1-|\alpha|/2}/ \alpha! \mbox{ by (\ref{E:CGF_Taylor})} \\
&= \frac{1}{2} t^T \Sigma t + \sum_{|\alpha| \ge 3}  \kappa_\alpha t^\alpha n^{1-|\alpha|/2}/ \alpha! \mbox{ since $\kappa_{e_i+e_j} = \Sigma_{ij}$} \\
&= \frac{1}{2} t^T \Sigma t + \sum_{r=3}^\infty T_r n^{1-r/2}
\end{align*}
where $T_r \defeq \sum_{|\alpha| = r}  \kappa_\alpha t^\alpha / \alpha!$ is the sum of all terms of degree $r$ in $t$.  So
\begin{align}
&M_{Z_n}(t)&& \nonumber \\
&= e^{t^T \Sigma t/2}&&\exp\left(\sum_{r=3}^\infty T_r n^{1-r/2} \right) \nonumber \\
&= e^{t^T \Sigma t/2}&&\left( 1 + \left[\sum_{r=3}^\infty T_r n^{1-r/2}\right]  + \frac{1}{2!} \left[\sum_{r=3}^\infty T_r n^{1-r/2}\right]^2  +
\frac{1}{3!} \left[\sum_{r=3}^\infty T_r n^{1-r/2}\right]^3
+ \dots  \right) \nonumber \\
&= e^{t^T \Sigma t/2}&&\left( 1 + n^{-1/2}\left[ T_3 \right] + n^{-1}\left[ T_4 + T_3^2/2 \right]
+ n^{-3/2}\left[ T_5 + T_3 T_4 + T_3^3/6 \right] \right. \nonumber \\
&&&+ n^{-2}\left[ T_6 + T_3 T_5 + T_4^2/2 + T_3^2 T_4 /2 + T_3^4/24 \right] \nonumber \\
&&&+ n^{-5/2}\left[ T_7 + T_3 T_6 + T_4 T_5 + T_3^2 T_5/2 + T_3 T_4^2/2 + T_3^3 T_4/6 + T_3^5/120 \right] \nonumber \\
&&& \left. + O(n^{-3})  \right) \label{E:MGF_approx}
\end{align}
where the second step uses the Taylor series expansion for the exponential function about $0$ (note that the argument is $O(1/\sqrt{n})$ so any desired accuracy can be obtained by truncating this series appropriately) and the third step simply collects together terms with the same power of $n$ (which is easy to do using a program which performs symbolic formula manipulation, such as Maxima \citep{maxima5.31.2}).

Now, $M_{Z_n}(t)$ is the integral transform
$$M_{Z_n}(t) = \E [e^{t\cdot {Z_n}}] = \int_{\R^d} e^{t\cdot y} \; d \zeta_n(y) $$
of the distribution $\zeta_n$ which are trying to approximate,
so we need to invert this transform.  But (\ref{E:MGF_approx}) approximates $M_{Z_n}(t)$ by a linear combination of terms which are each of the form $t^\alpha e^{t^T \Sigma t/2}$, which we will shortly show is the transform of $h_\alpha(y) \phi(y)$.
Using this property of the Hermite polynomials to invert (\ref{E:MGF_approx}) term-by-term then gives the Edgeworth expansion for $\zeta_n$.
For example, the $n^{-1}$ term in (\ref{E:MGF_approx}) is
$$ n^{-1} e^{t^T \Sigma t/2}\left[ T_4 + T_3^2/2 \right] =  n^{-1} e^{t^T \Sigma t/2} \left[
\sum_{|\alpha| = 4}  \frac{\kappa_\alpha t^\alpha}{\alpha!}  +
\sum_{|\alpha| = |\beta| = 3}  \frac{\kappa_\alpha \kappa_\beta t^{\alpha + \beta}}{2 \; \alpha! \beta!} \right] $$
whose inverse transform is therefore
$n^{-1} \phi(y) S_1(y)$, where $S_1(y)$ is given in the statement of the theorem.
The rigorous error analysis of \citep[Corollary 1.4]{BhattacharyaDenker90}
then proves the theorem.  So we finish by showing that the transform of $h_\alpha(y) \phi(y)$ is $t^\alpha e^{t^T \Sigma t/2}$:
\begin{align*}
&\int_{\R^d} e^{t\cdot y} h_\alpha(y) \phi(y) \; d\lambda(y) \nonumber \\
&= \int_{\R^d} e^{t\cdot y} (-1)^{|\alpha|} (\partial_y^\alpha \phi) \; d\lambda(y) \mbox{ by the definition (\ref{E:defn_Hermite}) of $h_\alpha(y)$} \\
&= \int_{\R^d} (\partial_y^\alpha e^{t\cdot y}) \phi(y) \; d\lambda(y) \mbox{ by integration by parts} \\
&= t^\alpha \int_{\R^d} e^{t\cdot y} \phi(y) \; d\lambda(y) \\
&= t^\alpha \int_{\R^d} (2\pi)^{-d/2} (\det \Sigma)^{-1/2} \exp( t\cdot y  - y^T \Sigma^{-1} y/2) \; d\lambda(y) \mbox{ by (\ref{E:PDF_normal})}\\
&= t^\alpha  \exp(t^T \Sigma t/2) \int_{\R^d} (2\pi)^{-d/2} (\det \Sigma)^{-1/2} \exp(-(y-\Sigma t)^T \Sigma^{-1} (y-\Sigma t)/2) \; d\lambda(y) \\
&= t^\alpha  e^{t^T \Sigma t/2},
\end{align*}
where the last step uses the fact that the integrand is the PDF for the normal distribution $N_d(\Sigma t, \Sigma)$ whose integral is $1$.
\end{proof}

When applying Theorem \ref{T:edgeworth} in the rest of this paper, our main case of interest will be when $y=0$.  We therefore now calculate the numbers $h_\alpha(0)$, which are known as the Hermite numbers.

\begin{theorem}
\label{T:Hermite_numbers}
The Hermite number $h_\alpha(0)$ is $0$ if $|\alpha|$ is odd and if $|\alpha|= 2k$ is even then
\begin{align}
h_\alpha(0)
&= \frac{(-1)^k}{k!} \; \partial_z^\alpha \left[\left(z^T \Sigma^{-1} z/2\right)^{k} \right]
\label{E:Hermite_numbers} \\
&= \frac{(-1)^k \alpha!}{k!} \sum_{\substack{|\beta^1| = \dots = |\beta^k|= 2 \\ \alpha = \beta^1+ \dots +\beta^k}}
\frac{\tilde{\kappa}_{\beta^1} \dots \tilde{\kappa}_{\beta^k}}{\beta^1! \dots \beta^k!},
\label{E:Hermite_numbers_pairs}
\end{align}
where the sum is over all ordered $k$-tuples $\beta^1, \dots, \beta^k$ of degree-$2$ multi-indices whose sum is $\alpha$, and  where $\tilde{\kappa}_{\beta^j} = \left( \Sigma^{-1} \right)_{ab}$ if $\beta^j = e_a + e_b$.
\end{theorem}

\begin{proof}
Recall that $\phi:\R^d \to \R$ is the PDF of $N_d(0,\Sigma)$ given by (\ref{E:PDF_normal})
and that the Hermite polynomial $h_\alpha(y)$ is defined by (\ref{E:defn_Hermite}), i.e., $h_\alpha \phi = (-1)^{|\alpha|} (\partial_y^\alpha \phi)$.
Using (\ref{E:defn_Hermite}), we see that our version of the Hermite polynomials satisfy the following recurrence relation, where we recall that $e_1, \dots, e_d$ is the standard basis for $\R^d$:
\begin{align}
h_{\alpha + e_i}
&= \phi^{-1} (-1)^{|\alpha + e_i|} (\partial_y^{\alpha + e_i} \phi)
= -\phi^{-1} \frac{\partial \;}{\partial y_i}  \left(h_\alpha \phi \right)
= -\frac{\partial h_\alpha}{\partial y_i} + h_\alpha h_{e_i}, \label{E:Hermite_recurrence}
\end{align}
where $\phi^{-1} = 1/\phi$.  It is easy to show from (\ref{E:PDF_normal}) and (\ref{E:defn_Hermite}) that $h_0(y)=1$ and $h_{e_i}(y) = e_i^T \Sigma^{-1} y$, so the Hermite polynomials can be calculated explicitly from (\ref{E:Hermite_recurrence}), though we will do not pursue this here.

Now, let $z \in \R^d$ be fixed, and define $\langle z,z \rangle = z^T \Sigma^{-1} z$ to be the square of the Mahalanobis norm.
Then for $s \in \R$ and any non-negative integer $r$,
\begin{align}
(-1)^r \frac{d^r}{d s^r} \left[ \phi(sz) \right]
&= (2\pi)^{-d/2} (\det \Sigma)^{-1/2} (-1)^r \frac{d^r}{d s^r} \left[\exp(-s^2 \langle z,z \rangle/2)\right] \mbox{ by (\ref{E:PDF_normal})} \nonumber \\
&= \langle z,z \rangle^{r/2}  \phi(sz) \tilde{h}_r\left(s \sqrt{\langle z,z \rangle}\right) \label{E:Hermite_calc1}
\end{align}
where $\tilde{h}_r$ is the $r^{th}$ $1$-dimensional, unit variance Hermite polynomial, defined by
$$ \tilde{h}_r(w) \exp(-w^2/2) = (-1)^r \frac{d^r}{d w^r} \exp(-w^2/2).$$
On the other hand, we will shortly use the chain rule to show that
\begin{align}
(-1)^r \frac{d^r}{d s^r} \left[ \phi(sz) \right]
&= \phi(sz) \sum_{|\alpha| = r}  n_\alpha h_\alpha(sz) z^\alpha    \label{E:Hermite_calc2}
\end{align}
where $n_\alpha = |\alpha|!/\alpha!$.  Combining (\ref{E:Hermite_calc1}) and (\ref{E:Hermite_calc2}) and evaluating at $s=0$ will therefore give
\begin{align}
\langle z,z \rangle^{r/2} \tilde{h}_r(0) = \sum_{|\alpha| = r}  n_\alpha h_\alpha(0) z^\alpha.  \label{E:Hermite_calc}
\end{align}
But it is well known that $\tilde{h}_r(0)$ is $0$ if $r$ is odd and $(-1)^{r/2}(r-1)!!$ if $r$ is even, where
$(r-1)!! = 1\times 3\times 5 \times \dots \times (r-1)$
is the double factorial.
So differentiating both sides of (\ref{E:Hermite_calc}) by $\partial_z^\alpha$
gives $h_\alpha(0) = 0$ if $|\alpha|$ is odd and (\ref{E:Hermite_numbers}) if $|\alpha|$ is even.  Then (\ref{E:Hermite_numbers_pairs}) follows from (\ref{E:Hermite_numbers}) and
$$ \frac{z^T \Sigma^{-1} z}{2} = \sum_{|\beta| = 2} \frac{\tilde{\kappa}_\beta z^\beta}{\beta!}. $$

So we finish by proving (\ref{E:Hermite_calc2}) by induction.  This is clearly true when $r=0$ since $h_0(y) = 1$, so we now assume (\ref{E:Hermite_calc2}) is true for $r = k$ and prove it for $r=k+1$.
\begin{align*}
(-1)^{k+1} \frac{d^{k+1}}{d s^{k+1}} \left[ \phi(sz) \right]
&= - \frac{d}{d s} \left[ \phi(sz) \sum_{|\alpha| = k}   n_\alpha h_\alpha(sz) z^\alpha \right] \mbox{ by the induction hypothesis} \\
&= - \sum_{|\alpha| = k} \sum_{i=1}^d z_i \left.\frac{\partial}{\partial y_i}\right|_{y=sz} \left[ \phi(y) h_\alpha(y)\right] n_\alpha z^\alpha \mbox{ by the chain rule}\\
&= \phi(sz) \sum_{|\alpha| = k} \sum_{i=1}^d n_\alpha z^{\alpha + e_i} \left[ h_{e_i}(sz) h_\alpha(sz) - \left.\frac{\partial h_\alpha}{\partial y_i}\right|_{y=sz} \right] \mbox{ by (\ref{E:defn_Hermite})}\\
&= \phi(sz) \sum_{|\alpha| = k} \sum_{i=1}^d n_\alpha z^{\alpha + e_i} h_{\alpha + e_i}(sz) \mbox{ by (\ref{E:Hermite_recurrence})}\\
&= \phi(sz) \sum_{|\beta| = k+1} n_\beta  h_\beta(sz) z^\beta
\end{align*}
since, for any $\beta$ with $|\beta| = k+1$, we have
\begin{align*}
\sum_{i:\beta_i \ge 1}  n_{\beta - e_i}
= \sum_{i:\beta_i \ge 1}  \frac{(|\beta|-1)!}{\beta_1! \dots (\beta_i- 1)! \dots \beta_d!}
= \frac{(|\beta|-1)!}{\beta!} \sum_{i:\beta_i \ge 1}  \beta_i
= \frac{|\beta|!}{\beta!} = n_\beta.
\end{align*}
\end{proof}

Multi-index notation is generally more concise than tensor notation for expressions involving multi-indices $\alpha$ where $|\alpha|$ is large, but not when $|\alpha|$ is small.  Also, many people are uncomfortable with multi-index notation, despite its advantages, so we finish this section by converting some of the expressions above into tensor notation.
The key identity for this conversion is
\begin{align}
\label{E:conversion_appendix}
\sum_{|\alpha|=r}  \frac{f(\alpha)}{\alpha!} = \frac{1}{r!} \sum_{i_1, \dots, i_r} f(e_{i_1} + \dots + e_{i_r}),
\end{align}
where $f$ is any function of degree $r$ multi-indices.  This formula follows from
$$ \sum_{i_1, \dots, i_r} f(e_{i_1} + \dots + e_{i_r})
= \sum_{|\alpha|=r} \sum_{\substack{i_1, \dots, i_r \\ \alpha = e_{i_1} + \dots + e_{i_r}}} f(\alpha)
= \sum_{|\alpha|=r} f(\alpha) \sum_{\substack{i_1, \dots, i_r \\ \alpha = e_{i_1} + \dots + e_{i_r}}} 1
= \sum_{|\alpha|=r} \frac{f(\alpha) r!}{\alpha!},
$$
where the last step uses the fact that $r!/\alpha!$ is the multinomial coefficient which counts the number of ways of putting $r$ distinguishable objects into $d$ boxes so that $\alpha_i$ objects go into box $i$.

For any indices $i_1, \dots, i_r \in \{1, \dots, d\}$, let $\kappa_{i_1 \dots i_r}$ and $h_{i_1 \dots i_r}$ be alternative notation for $\kappa_\alpha$ and $h_\alpha(0)$, respectively, where $\alpha = e_{i_1} + \dots + e_{i_r}$.  Then (\ref{E:conversion_appendix}) implies that $S_1(y)$ and $S_2(y)$ in Theorem \ref{T:edgeworth} are given at $y=0$ by
\begin{align}
\label{E:conversion_S1_appendix}
S_1(0) = \frac{1}{24} \sum_{i_1, \dots, i_4}  \kappa_{i_1 \dots i_4} h_{i_1 \dots i_4}  +
\frac{1}{72} \sum_{i_1, \dots, i_6}  \kappa_{i_1 i_2 i_3} \kappa_{i_4 i_5 i_6} h_{i_1 \dots i_6}
\end{align}
and
\begin{align}
S_2(0) =& \frac{1}{720} \sum_{i_1, \dots, i_6}  \kappa_{i_1 \dots i_6} h_{i_1 \dots i_6}  +
\frac{1}{720} \sum_{i_1, \dots, i_8}  \kappa_{i_1 i_2 i_3} \kappa_{i_4 \dots i_8} h_{i_1 \dots i_8} + \nonumber \\
& \frac{1}{1152} \sum_{i_1, \dots, i_8}  \kappa_{i_1 \dots i_4} \kappa_{i_5 \dots i_8} h_{i_1 \dots i_8}
+ \frac{1}{1728} \sum_{i_1, \dots, i_{10}}  \kappa_{i_1 i_2 i_3} \kappa_{i_4 i_5 i_6} \kappa_{i_7 \dots i_{10}} h_{i_1 \dots i_{10}} + \nonumber  \\
& \frac{1}{31104} \sum_{i_1, \dots, i_{12}}  \kappa_{i_1 i_2 i_3} \kappa_{i_4 i_5 i_6} \kappa_{i_7 i_8 i_9} \kappa_{i_{10} i_{11} i_{12}} h_{i_1 \dots i_{12}}.  \label{E:conversion_S2_appendix}
\end{align}
Also, by Theorem \ref{T:Hermite_numbers} and (\ref{E:conversion_appendix}), we have
\begin{align}
\label{E:conversion_h4_appendix}
h_{i_1 \dots i_4} = g^{i_1i_2}g^{i_3i_4} + g^{i_1i_3}g^{i_2i_4} + g^{i_1i_4}g^{i_2i_3}
\end{align}
and
\begin{align}
\label{E:conversion_h6_appendix}
h_{i_1 \dots i_6} = -(&
g^{i_1i_2}g^{i_3i_4}g^{i_5i_6} + g^{i_1i_3}g^{i_2i_4}g^{i_5i_6} + g^{i_1i_4}g^{i_2i_3}g^{i_5i_6} + g^{i_1i_2}g^{i_3i_5}g^{i_4i_6} + \nonumber \\
&g^{i_1i_3}g^{i_2i_5}g^{i_4i_6} + g^{i_1i_5}g^{i_2i_3}g^{i_4i_6} + g^{i_1i_2}g^{i_3i_6}g^{i_4i_5} + g^{i_1i_3}g^{i_2i_6}g^{i_4i_5} + \nonumber \\
&g^{i_1i_6}g^{i_2i_3}g^{i_4i_5} + g^{i_1i_4}g^{i_2i_5}g^{i_3i_6} + g^{i_1i_5}g^{i_2i_4}g^{i_3i_6} + g^{i_1i_4}g^{i_2i_6}g^{i_3i_5} + \nonumber \\
&g^{i_1i_6}g^{i_2i_4}g^{i_3i_5} + g^{i_1i_5}g^{i_2i_6}g^{i_3i_4} + g^{i_1i_6}g^{i_2i_5}g^{i_3i_4} )
\end{align}
where $g^{ij} = (\Sigma^{-1})_{ij}$.  These all agree with \citep[\S 4.5]{KassVos97} to the extent that our equations and theirs overlap.

\bibliographystyle{abbrvnat}
\bibliography{R:/5050/CEB/Share/Staff/EnesAndDaniel/Bibliography/bibliography}

\begin{thebibliography}{19}
\providecommand{\natexlab}[1]{#1}
\providecommand{\url}[1]{\texttt{#1}}
\expandafter\ifx\csname urlstyle\endcsname\relax
  \providecommand{\doi}[1]{doi: #1}\else
  \providecommand{\doi}{doi: \begingroup \urlstyle{rm}\Url}\fi

\bibitem[Amari and Nagaoka(2000)]{AmariNagaoka00}
S.~Amari and H.~Nagaoka.
\newblock \emph{Methods of Information Geometry}, volume 191 of
  \emph{Translations of mathematical monographs}.
\newblock American Mathematical Society, 2000.

\bibitem[Barndorff-Nielsen(1978)]{BarndorffNielsen78}
O.~Barndorff-Nielsen.
\newblock \emph{Information and exponential families}.
\newblock John Wiley \& Sons, 1978.

\bibitem[Barndorff-Nielsen and Cox(1979)]{BarndorffNielsenCox79}
O.~Barndorff-Nielsen and D.~R. Cox.
\newblock Edgeworth and saddle-point approximations with statistical
  applications.
\newblock \emph{Journal of the Royal Statistical Society Series B
  (Methodological)}, 41\penalty0 (3):\penalty0 279--312, 1979.

\bibitem[Barron and Cover(1991)]{BarronCover91}
A.~R. Barron and T.~M. Cover.
\newblock Minimum complexity density estimation.
\newblock \emph{{IEEE} Transactions on Information Theory}, 37\penalty0
  (4):\penalty0 1034--1054, July 1991.

\bibitem[Barron et~al.(1998)Barron, Rissanen, and Yu]{BarronEtAl98}
A.~R. Barron, J.~Rissanen, and B.~Yu.
\newblock The minimum description length principle in coding and modeling.
\newblock \emph{{IEEE} Transactions on Information Theory}, 44\penalty0
  (6):\penalty0 2743--2760, October 1998.

\bibitem[Bhattacharya and Denker(1990)]{BhattacharyaDenker90}
R.~Bhattacharya and M.~Denker.
\newblock \emph{Asymptotic Statistics}, volume~14 of \emph{Oberwolfach
  Seminars}.
\newblock Birkh\"auser Basel, 1990.

\bibitem[Clarke and Barron(1990)]{ClarkeBarron90}
B.~S. Clarke and A.~R. Barron.
\newblock Information-theoretic asymptotics of {B}ayes methods.
\newblock \emph{{IEEE} Transactions on Information Theory}, 36\penalty0
  (3):\penalty0 453--471, May 1990.

\bibitem[Clarke and Barron(1994)]{ClarkeBarron94}
B.~S. Clarke and A.~R. Barron.
\newblock {J}effreys' prior is asymptotically least favorable under entropy
  risk.
\newblock \emph{Journal of Statistical Planning and Inference}, 41\penalty0
  (1):\penalty0 37 -- 60, 1994.
\newblock ISSN 0378-3758.

\bibitem[Efron(1975)]{Efron75}
B.~Efron.
\newblock Defining the curvature of a statistical problem (with applications to
  second order efficiency).
\newblock \emph{Ann. Statist.}, 3\penalty0 (6):\penalty0 1189--1242, 11 1975.

\bibitem[Gray(1973)]{Gray73}
A.~Gray.
\newblock The volume of a small geodesic ball of a {R}iemannian manifold.
\newblock \emph{Michigan Math. J.}, 20:\penalty0 329--344, 1973.

\bibitem[Gr\"unwald(2005)]{Grunwald05}
P.~Gr\"unwald.
\newblock A tutorial introduction to the minimum description length principle.
\newblock In I.~J.~M. P.~Gr\"unwald and M.~Pitt, editors, \emph{Advances in
  Minimum Description Length: Theory and Applications}. MIT Press, 2005.

\bibitem[Hall(1992)]{Hall92}
P.~Hall.
\newblock \emph{The bootstrap and {E}dgeworth expansions}.
\newblock Springer, New York, 1992.

\bibitem[Kass and Vos(1997)]{KassVos97}
R.~E. Kass and P.~W. Vos.
\newblock \emph{Geometrical Foundations of Asymptotic Inference}.
\newblock John Wiley \& Sons, 1997.

\bibitem[Nelson(1973)]{Nelson73}
E.~Nelson.
\newblock Probability theory and {E}uclidean field theory.
\newblock In G.~Velo and A.~Wightman, editors, \emph{Constructive Quantum Field
  Theory}, volume~25 of \emph{Lecture Notes in Physics}, pages 94--124.
  Springer Berlin Heidelberg, 1973.

\bibitem[Rissanen(1996)]{Rissanen96}
J.~Rissanen.
\newblock Fisher information and stochastic complexity.
\newblock \emph{{IEEE} Transactions on Information Theory}, 42\penalty0
  (1):\penalty0 40--47, January 1996.

\bibitem[Rissanen(2000)]{Rissanen00}
J.~Rissanen.
\newblock {MDL} denoising.
\newblock \emph{{IEEE} Transactions on Information Theory}, 46\penalty0
  (7):\penalty0 2537--2543, 2000.

\bibitem[Rissanen(2007)]{Rissanen07}
J.~Rissanen.
\newblock \emph{Information and Complexity in Statistical Modeling}.
\newblock Information Science and Statistics. Springer, first edition, 2007.

\bibitem[Schelter(2015)]{maxima5.31.2}
W.~Schelter.
\newblock \emph{Maxima, a Computer Algebra System. Version 5.31.2}.
\newblock Macsyma group, 2015.
\newblock URL \url{http://maxima.sourceforge.net}.

\bibitem[Shtarkov(1987)]{Shtarkov87}
Y.~M. Shtarkov.
\newblock Universal sequential coding of single messages.
\newblock \emph{Probl. Inform. Transm.}, 23\penalty0 (3):\penalty0 3--17, 1987.

\end{thebibliography}

\end{document}